\documentclass[12pt]{llncs}
\usepackage{amsmath,amssymb}
\usepackage{graphicx}
\usepackage{fullpage}

\def\bbR{\mathbb{R}}
\newtheorem{observation}{Observation}

\begin{document}

\title{A Linear-Time Algorithm for Radius-Optimally Augmenting Paths in a Metric Space\thanks{A preliminary version of this paper will appear in the Proceedings of the 16th Algorithms and Data Structures Symposium (WADS 2019).}}
\author{
Christopher Johnson
\and
Haitao Wang
}
\institute{
Department of Computer Science\\
Utah State University, Logan, UT 84322, USA\\
\email{christopherajohnson42@gmail.com, haitao.wang@usu.edu}\\
}

\maketitle

\pagestyle{plain}
\pagenumbering{arabic}
\setcounter{page}{1}

\begin{abstract}
Let $P$ be a path graph of $n$ vertices embedded in a metric space. We consider the problem of adding a new edge to $P$ to minimize the radius of the resulting graph. Previously, a similar problem for minimizing the diameter of the graph was solved in $O(n\log n)$ time. To the best of our knowledge, the problem of minimizing the radius has not been studied before. In this paper, we present an $O(n)$ time algorithm for the problem, which is optimal.
\end{abstract}

\section{Introduction}
\label{sec:intro}

In this paper, we consider the problem of augmenting a path graph embedded in a metric space by adding a new edge so that the radius of the new graph is minimized.

Let $P$ be a path graph of $n$ vertices, $v_1,v_2,\ldots,v_n$, ordered from one end to the other. Let $e(v_i,v_{i+1})$ denote the edge connecting two vertices $v_{i}$ and $v_{i+1}$ for $i\in [1,n-1]$.
Let $V$ be the set of all vertices of $P$. We assume that $P$ is embedded in a metric space, i.e., $(V,|\cdot|)$ is a metric space and $|v_iv_j|$ is the distance of any two vertices $v_i$ and $v_j$ of $V$.
Specifically, the following properties hold: (1) the triangle inequality:
$|v_iv_k|+|v_kv_j|\geq |v_iv_j|$; (2) $|v_iv_j|=|v_jv_i|\geq 0$; (3) $|v_iv_j|=0$ iff $i=j$.
For each edge $e(v_{i},v_{i+1})$ of $P$, its {\em length} is equal to
$|v_{i}v_{i+1}|$.

Suppose we add a new edge $e$ connecting two vertices $v_i$ and $v_j$ of $P$, and let $P\cup \{e\}$ denote the resulting graph. Note that a {\em point} of $P$ can be either a vertex of $P$ or in the interior of an edge.
A point $c$ on $P\cup \{e\}$ is called a {\em center} if it minimizes
the largest shortest path length from $c$ to all vertices of $P$, and
the largest shortest path length from the center to all vertices is called the {\em radius}.
Our problem is to add a new edge $e$ to connect two vertices of $P$ such that the radius of $P\cup \{e\}$ is minimized. We refer to the problem as {\em the radius-optimally augmenting path
problem}, or ROAP for short.

To the best of our knowledge, the problem has not been studied before.
In this paper, we present an $O(n)$ time algorithm. We assume that the distance $|v_iv_j|$ can be obtained in $O(1)$ time for any two vertices $v_i$ and $v_j$ of $P$.

As a by-product of our techniques, we
present an algorithm that can compute the radius and the center of
$P\cup\{e\}$  in $O(\log n)$ time for any given new edge $e$, after $O(n)$ time
preprocessing.


\subsection{Related Work}
A similar problem for minimizing the diameter of the augmenting graph was studied before. Gro{\ss}e et
al.~\cite{ref:GrobeFa15} first gave an $O(n\log^3 n)$ time algorithm, and later Wang~\cite{ref:WangAn18} solved the problem in $O(n\log n)$ time.

Some variations of the diameter problem have also been considered in the literature.
If the path $P$ is in the Euclidean space $\bbR^d$ for a constant $d$, then Gro{\ss}e et
al.~\cite{ref:GrobeFa15} gave an $O(n+1/\epsilon^3)$ time
algorithm that can find a $(1+\epsilon)$-approximate solution for
the diameter problem, for any $\epsilon>0$. If $P$ is in the Euclidean plane $\bbR^2$,
De Carufel et al.~\cite{ref:DeCarufelMi16} gave a linear time algorithm for
adding a new edge to $P$ to minimize the {\em continuous diameter}
(i.e., the diameter is defined with respect to all points of $P$, not only vertices).
For a geometric tree $T$ of $n$ vertices embedded in the Euclidean plane, De Carufel et al.~\cite{ref:DeCarufelMi17} gave an $O(n\log n)$ time algorithm for
adding a new edge to $T$ to minimize the {\em continuous} diameter. For the discrete diameter problem where $T$ is embedded in a metric space, Gro{\ss}e et al.~\cite{ref:GrobeFa16} first proposed an $O(n^2\log n)$ time algorithm and later Bil{\`o}~\cite{ref:BiloAl18} solved the problem in $O(n\log n)$ time.
Oh and Ahn~\cite{ref:OhA16} studied the problem on a general tree (i.e., the tree is not embedded in a metric space) and gave $O(n^2\log^3 n)$ time algorithms for both the discrete and continuous versions of the diameter problem, and later Bil{\`o}~\cite{ref:BiloAl18} gave an improved algorithm of $O(n^2)$ time for the discrete case, which is optimal.

The more general problem of adding $k$ edges to a graph $G$ so that the diameter of the resulting graph is minimized has also been considered before. The problem is NP-hard \cite{ref:SchooneDi97} and some other variants are even W[2]-hard \cite{ref:FratiAu15,ref:GaoTh13}. Approximation algorithms have been proposed~\cite{ref:BiloIm12,ref:FratiAu15,ref:LiOn92}.
The upper and lower bounds on the diameters
of the augmented graphs were also studied, e.g., \cite{ref:AlonDe00,ref:IshiiAu13}.
Bae {\em et al.}~\cite{ref:BaeSh17} considered the problem of adding $k$ shortcuts to a circle in the plane to minimize the diameter of the resulting graph.


Like the diameter, the radius is a critical metric of network performance, which
measures the worst-case cost between a ``center'' and all other nodes. Therefore, our problem of
augmenting graphs to minimize the radius potentially has many applications.
As an example, suppose there is a highway that connects several cities and we want to build a facility along the highway to provide certain service for all these cities. In order to reduce the transportation time, we plan to build a new highway connecting two cities such that the radius (i.e., the maximum distance from the cities to the facility located at the center) is as small as possible.

\subsection{Our Approach}
Note that in general the radius of $P\cup \{e\}$
is not equal to the diameter divided by two. For example, suppose
$e$ connects $v_1$ and $v_n$ (i.e., $P\cup \{e\}$ is a cycle).
Assume that the edges of the cycle have the same length and $n$ is
even. Suppose the total length of the cycle is $1$. Then, the diameter of the
cycle is $1/2$ while the radius is $(1-1/n)/2$, very close to the diameter.

\begin{figure}[t]
\begin{minipage}[t]{\textwidth}
\begin{center}
\includegraphics[height=1.3in]{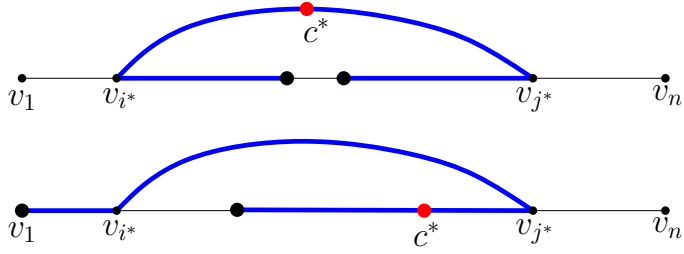}
\caption{\footnotesize Illustrating two configurations for the optimal solution, where $c^*$ is the center and the thick (blue) paths are shortest paths from $c^*$ to its two farthest vertices, depicted by larger points. In the top configuration,
$c^*$ is on the new edge $e$ and both farthest vertices are on the sub-path of $P$ between $v_{i^*}$ and $v_{j^*}$. In the bottom configuration, $c^*$ is on the sub-path of $P$ between $v_{i^*}$ and $v_{j^*}$; $v_1$ is a farthest vertex of $c^*$ and the other one is on the sub-path of $P$ between $v_{i^*}$ and $v_{j^*}$.
There are also other configurations, e.g.,
$c^*$ is on the sub-path of $P$ between $v_1$ and $v_{i^*}$.}
\label{fig:config}
\end{center}
\end{minipage}
\vspace{-0.15in}
\end{figure}

One straightforward way to solve the problem ROAP is to try all edges
$e$ connecting $v_i$ and $v_j$ for all $i,j\in [1,n]$, which would
take $\Omega(n^2)$ time. We instead use the following approach. Suppose an optimal edge $e$ connecting two vertices $v_{i^*}$ and $v_{j^*}$. Depending on the locations of the center $c^*$ and its two farthest vertices in $P\cup \{e\}$, there are several possible {\em configurations} for the optimal solution (e.g., see Fig.~\ref{fig:config}).
For each such configuration, we compute the best solution for it in
linear time, such that if there is an optimal solution conforming with the configuration, our solution is also optimal. The efficiency of our approach relies on many observations and certain monotonicity properties, which help us avoid the brute-force method. In fact, our algorithm, which involves several kinds of linear scans,  is relatively simple. The challenge, however, is on discovering and proving these observations and properties. To this end, our main tool is the triangle inequality of the metric space.

\paragraph{Outline.} The remaining paper is organized as follows.
In Section~\ref{sec:pre}, we introduce some notation.
In Section~\ref{sec:algo}, we present our linear time algorithm for ROAP.
In Section~\ref{sec:query}, we discuss our $O(\log n)$ time query algorithm for computing the radius and the center of $P\cup \{e\}$.

\section{Preliminaries}
\label{sec:pre}


Denote by $e(v_i,v_j)$  the edge connecting two vertices
$v_i$ and $v_j$ for any $i,j\in [1,n]$. The length of
$e(v_i,v_j)$ is $|v_iv_j|$.
This implies that for any two points $p$ and $q$ on $e(v_i,v_j)$, the length of the portion of $e(v_i,v_j)$ between $p$ and $q$ is $|pq|$.
Later we will use this property directly without further explanations.

%

For any two points $p$ and $q$ on $P$, we use $P(p,q)$ to denote the
subpath of $P$ between $p$ and $q$.
Unless otherwise stated, we assume $i\leq j$ for each index pair
$(i,j)$ discussed in the paper.  For any pair $(i,j)$, we use $G(i,j)$ to
denote the new graph $P\cup \{e(v_i,v_j)\}$ and use $C(i,j)$ to denote the cycle
$P(v_i,v_j)\cup e(v_i,v_j)$.
Note that if $j\leq i+1$, then $G(i,j)=P$ and $C(i,j)=P(v_i,v_j)$.

For any graph $G$, we use $d_{G}(p,q)$ to denote the length of the shortest
path between two points $p$ and $q$ in $G$, and we also call
$d_{G}(p,q)$ the {\em distance} between $p$ and $q$ in $G$. In our
paper, $G$ is usually a subgraph of $G(i,j)$, e.g., $P$ or $C(i,j)$.
For example, $d_P(p,q)$ denotes the length of $P(p,q)$.
We perform a linear time preprocessing so that $d_P(v_i,v_j)$ can be computed in $O(1)$ time for any pair of $(i,j)$. Recall that $|v_iv_j|>0$ unless $i=j$, and thus, $d_P(v_i,v_j)>0$ unless $i=j$.

A {\em center} of $G(i,j)$ is defined as a point (which can be either a vertex or in the interior of an edge) that minimizes the maximum distance from it to all vertices in $G(i,j)$\footnote{The concept of center is defined with respect to the graph instead of to the metric space.}, and the maximum distance is called the {\em radius} of $G(i,j)$. Hence, the problem ROAP is to find a pair of
indices $(i,j)$ such that the radius of $G(i,j)$ is
minimized.

For convenience, we assume that $P$ from $v_1$ to $v_n$ is oriented
from left to right, so that we can talk about the relative positions of
the points of $P$ (i.e., a point $p$ is to the left of another point
$q$ on $P$ if $p$ is closer to $v_1$ than $q$ is).
Similarly, each edge $e(v_i,v_j)$ with $i<j$ from $v_i$ to $v_j$
is oriented from left to right.

\section{Our Algorithm for ROAP}
\label{sec:algo}

In this section, we present our algorithm for solving the problem ROAP. Let $(i^*,j^*)$ be an
optimal solution with $i^*\leq j^*$ and $c^*$ be a center of $G(i^*,j^*)$. Let $r^*$
denote the radius of $G(i^*,j^*)$. We begin with the following observation.

\begin{observation}\label{obser:10}
In $G(i^*,j^*)$, there are two vertices $v_{a^*}$ and $v_{b^*}$ such that the following are true.
\begin{enumerate}
\item
$d_{G(i^*,j^*)}(c^*,v_{a^*})=d_{G(i^*,j^*)}(c^*,v_{b^*})=r^*$.
\item
There is a shortest path from $c^*$ to $v_{a^*}$, denoted by $\pi_{a^*}$, and a shortest path from $c^*$ to $v_{b^*}$, denoted by $\pi_{b^*}$, such that $c^*$ is at the middle of $\pi_{a^*}\cup \pi_{b^*}$  (i.e., the concatenation of the two paths).
\end{enumerate}
\end{observation}
\begin{proof}
If this were not true, then we could slightly move $c^*$ so that the maximum distance from the new position of $c^*$ to all vertices in $G(i^*,j^*)$ becomes smaller than $r^*$, which contradicts with the definition of $r^*$. \qed
\end{proof}

Let $a^*$ and $b^*$ be the indices of the two vertices $v_{a^*}$ and $v_{b^*}$,  and $\pi^*$ be the union of the two paths $\pi_{a^*}$ and $\pi_{b^*}$ stated in Observation~\ref{obser:10}. Without loss of generality, we assume  $a^*< b^*$.

Depending on the locations of $c^*$, $a^*$, $b^*$, as well as whether
$e(v_{i^*},v_{j^*})\in \pi^*$, there are several possible
configurations. For each configuration, we will give a linear time
algorithm to compute a candidate solution, i.e., a pair $(i,j)$ (along with a radius $r$ and a center $c$), so that if there is an optimal solution conforming with the configuration then $(i,j)$ is also an optimal solution with $c$ as the center and $r=r^*$. On the other hand, each such solution is {\em feasible} in the sense that the distances from $c$ to all vertices in $G(i,j)$ is at most $r$. There are a constant number of configurations. Since we do not know which configuration has an optimal solution, we will compute a candidate solution for each configuration, and finally, among all candidate solutions we return the one with the smallest radius. The running time of the algorithm is $O(n)$.

For example, one
configuration is that $a^* = 1$ and $c^*$ is on $P(v_1,v_{i^*})$. In
this case, $r^*$ is equal to $d_P(v_1,c^*)$ and also equal to
$d_P(c^*,v_{i^*})$ plus the distance from $v_{i^*}$ to its
farthest vertex $v_k$ for all $k\in [i^*,n]$, i.e., $\max_{k\in [i^*,n]}d_{G(i^*,j^*)}(v_{i^*},v_k)$.
In other words, $r^*$ is equal to half of
$d_P(v_1,v_{i^*})+\max_{k\in [i^*,n]}d_{G(i^*,j^*)}(v_{i^*},v_k)$.
Further, it can be verified that $j^*$ must be the index $j$ that
minimizes the value $\max_{k\in [i^*,n]}d_{G(i^*,j)}(v_{i^*},v_k)$ among all $j\in
[i^*,n]$. Therefore, $r^*$ is equal to half of
$d_P(v_1,v_{i^*})+\min_{j\in [i^*,n]}\max_{k\in
[i^*,n]}d_{G(i^*,j)}(v_{i^*},v_k)$. Also, since $c^*\in P(v_1,v_{i^*})$,
$d_P(v_1,v_{i^*})\geq \max_{k\in [i^*,n]}d_{G(i^*,j^*)}(v_{i^*},v_k)$.
Correspondingly, we can compute a candidate solution as follows.

For any $i\in [1,n]$, define $\lambda_i=\min_{j\in [i,n]}\max_{k\in
[i,n]}d_{G(i,j)}(v_{i},v_k)$, and let $j(i)$ denote the index $j\in [i,n]$ that
achieves $\lambda_i$. Suppose $\lambda_i$ and $j(i)$ for all $i\in
[1,n]$ are known. Then, in $O(n)$ time we can find the index $i$ that
minimizes the value $d_P(v_1,v_i)+\lambda_i$ among all $i\in [1,n]$
with $d_P(v_1,v_i)\geq \lambda_i$. We return the pair $(i,j(i))$ (with radius $r=(d_P(v_1,v_i)+\lambda_i)/2$ and center $c$ as the point on $P(v_1,v_i)$ such that $d_P(v_1,c)=r$) as the
candidate solution for the configuration.
It is not difficult to see that if the configuration has an optimal
solution, then $(i,j(i))$ is an optimal solution with the center at $c$ and $r^*=r$.
Further, by our definition of $\lambda_i$, the distance from $c$ to every vertex  in $G(i,j(i))$ is at most $r$, and thus our candidate solution is feasible.

According to the above discussion, we need to compute  $\lambda_i$ and
$j(i)$ for all $i\in [1,n]$, which is done in the following lemma, with proof in Section~\ref{sec:lem10}.

\begin{lemma}\label{lem:10}
There is an algorithm that can compute $\lambda_i$ and $j(i)$ in $O(n)$ time for all $i\in [1,n]$.
\end{lemma}

\subsection{The Algorithm for Lemma~\ref{lem:10}}
\label{sec:lem10}

The success of our approach hinges on several monotonicity properties
that we shall prove.

Consider any $i\in [1,n]$.
For any $j\in [i,n]$, define $\alpha(i,j)=\max_{k\in
[i,n]}d_{G(i,j)}(v_{i},v_k)$, $\beta(i,j)=\max_{k\in
[i,j]}d_{G(i,j)}(v_{i},v_k)$, and $\gamma(i,j)=\max_{k\in
[j+1,n]}d_{G(i,j)}(v_{i},v_k)$ if $j<n$ and $\gamma(i,j)=0$ otherwise. Clearly, $\alpha(i,j)=\max\{\beta(i,j),\gamma(i,j)\}$ and $\lambda_i=\min_{j\in
[i,n]}\alpha(i,j)$.

Note that for any $k\in [i,j]$, the shortest path from $v_i$ to $v_k$ in $G(i,j)$ must be in the cycle $C(i,j)$. Hence, $\beta(i,j)=\max_{k\in [i,j]}d_{C(i,j)}(v_{i},v_k)$. Also, it is not difficult to see that $\gamma(i,j)= d_{G(i,j)}(v_{i},v_n)$. For $d_{G(i,j)}(v_{i},v_n)$, there are two paths from $v_i$ to $v_n$ in $G(i,j)$: $P(v_i,v_n)$ and $e(v_i,v_j)\cup P(v_j,v_n)$. The length of the latter path is $|v_iv_j|+d_P(v_j,v_n)$. Due to the triangle inequality in the metric space, it holds that $|v_iv_j|\leq d_P(v_i,v_j)$. Hence,
$\gamma(i,j)=|v_iv_j|+d_P(v_j,v_n)$.
Our first monotonicity property is given in the following
lemma, which is due to the triangle inequality.

\begin{lemma}\label{lem:20}
	$\gamma(i,j)\geq \gamma(i,j+1)$ for all $j\in [i,n-1]$.
\end{lemma}
\begin{proof}
Since $\gamma(i,j)=|v_iv_j|+d_P(v_j,v_n)$ and
	$\gamma(i,j+1)=|v_iv_{j+1}|+d_P(v_{j+1},v_n)$, we have
	$\gamma(i,j)-\gamma(i,j+1)=|v_iv_{j}|+d_P(v_j,v_{j+1})-|v_iv_{j+1}|
	=|v_iv_{j}|+|v_jv_{j+1}|-|v_iv_{j+1}|\geq 0$. The last inequality
	is due to the triangle inequality. \qed
\end{proof}

Let $I(i,j)$ be the index $k$ in $[i,j]$ such that
$\beta(i,j)=d_{C(i,j)}(v_i,v_{k})$ (if there more than
one such $k$, then we let $I(i,j)$ refer to the smallest one). In our algorithm given later, we will need to compute $\beta(i,j)$ for some pairs $(i,j)$.
For each $k\in [i,j]$, observe that
	$d_{C(i,j)}(v_i,v_k)=\min\{d_P(v_i,v_k),|v_iv_j|+d_P(v_k,v_j)\}$.
Hence, if $I(i,j)$ is known, then $\beta(i,j)$ can be computed in constant time due to our preprocessing in Section~\ref{sec:pre}.
In order to determine $I(i,j)$, we introduce a new notation.
Define $I'(i,j)$ to be the smallest index $k\in [i,j]$ such that $d_P(v_i,v_k)\geq |v_iv_j|+d_P(v_k,v_j)$. Note that such $k$ must exist since $d_P(v_i,v_j)\geq |v_iv_j|$.

\begin{observation}\label{obser:20}
$I(i,j)$ is either $I'(i,j)$ or $I'(i,j)-1$.
\end{observation}
\begin{proof}
Let $h=I'(i,j)$. We first assume that $h>i$. By the definition of $I'(i,j)$, $d_P(v_i,v_h)\geq |v_iv_j|+d_P(v_j,v_h)$ and $d_P(v_i,v_{h-1})< |v_iv_j|+d_P(v_j,v_{h-1})$. Thus, $d_{C(i,j)}(v_i,v_h)=|v_iv_j|+d_P(v_j,v_h)$ and $d_{C(i,j)}(v_i,v_{h-1})=d_P(v_i,v_{h-1})$.

Consider any $k\in [i,j]$.
If $k> h$, then $d_{C(i,j)}(v_i,v_h)=|v_iv_j|+d_P(v_j,v_h)\geq |v_iv_j|+d_P(v_j,v_k)\geq d_{C(i,j)}(v_i,v_k)$.
If $k<h-1$, then $d_{C(i,j)}(v_i,v_{h-1})=d_P(v_i,v_{h-1})>d_P(v_i,v_k)\geq d_{C(i,j)}(v_i,v_k)$. Note that $d_P(v_i,v_{h-1})>d_P(v_i,v_k)$ holds because $d_P(v_i,v_{h-1})=d_P(v_i,v_k)+d_P(v_k,v_{h-1})$ and $d_P(v_k,v_{h-1})>0$. Therefore, one of $h$ and $h-1$ must be $I(i,j)$.

If $h=i$, since $d_P(v_i,v_h)=0$, by the definition of $h$, it must be the case that $i=j$. Thus, $I(i,j)=i=h$.
\qed
\end{proof}

Observation~\ref{obser:20} tells that if we know $I'(i,j)$, then $\beta(i,j)$ is equal to the minimum of $d_{C(i,j)}(v_i,v_{I'(i,j)})$ and $d_{C(i,j)}(v_i,v_{I'(i,j)-1})$, which can be computed in constant time. Hence, to compute $\beta(i,j)$, it is sufficient to determine $I'(i,j)$. To efficiently compute $I'(i,j)$ during our algorithm, the following monotonicity properties on $I'(i,j)$ will be quite helpful.

\begin{lemma}\label{lem:30}
	\begin{enumerate}
		\item
			$I'(i,j)\leq I'(i,j+1)$ for all $j\in [i,n-1]$.
		\item
			$I'(i,j)\leq I'(i+1,j)$ for all $i\in [1,j-1]$.
	\end{enumerate}
\end{lemma}
\begin{proof}
	Recall that $d_{C(i,j)}(v_i,v_k)=\min\{d_P(v_i,v_k),|v_iv_j|+d_P(v_j,v_k)\}$ for each $k\in [i,j]$. If we
	consider $k$ as an integral variable changing from $i$ to $j$,
	then $d_P(v_i,v_k)$ is monotonically increasing and
	$|v_iv_j|+d_P(v_j,v_k)$ is monotonically decreasing (e.g., see Fig.~\ref{fig:fundC}). Hence, in
	general $d_{C(i,j)}(v_i,v_k)$ is a unimodal function of $k\in [i,n]$, i.e., first
	increases and then decreases.
    Therefore, we have an easy {\em observation}:
	If $d_P(v_i,v_h)\geq |v_iv_j|+d_P(v_h,v_j)$ for some $h\in [i,j]$, then $I'(i,j)\leq h$.

\begin{figure}[t]
\begin{minipage}[t]{\textwidth}
\begin{center}
\includegraphics[height=1.3in]{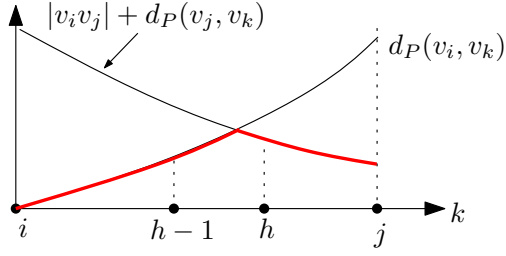}
\caption{\footnotesize Illustrating the two functions $d_P(v_i,v_k)$ and $|v_iv_j|+d_P(v_j,v_k)$ for $k\in [i,j]$, where $d_{C(i,j)}(v_i,v_k)$ is the pointwise minimum of them, depicted by thick (red) curve.
Roughly speaking, $I'(i,j)$ is the first index (i.e., $h$ in the figure) after the intersection of the two functions.}
\label{fig:fundC}
\end{center}
\end{minipage}
\vspace{-0.15in}
\end{figure}

For the first part of the lemma, let $h=I'(i,j+1)$. If $h=j+1$, then it is obviously true that $I'(i,j)\leq h$. Otherwise, $h\in [i,j]$. By definition, $d_P(v_{i},v_h)\geq |v_{i}v_{j+1}|+d_P(v_{j+1},v_h)$.
 Observe that $|v_iv_{j+1}|+d_P(v_{j+1},v_h)\geq |v_iv_j|+d_P(v_{j},v_h)$ due to
	the triangle inequality. Hence, $d_P(v_{i},v_h)\geq |v_iv_j|+d_P(v_j,v_h)$. By the above observation, $I'(i,j)\leq h$.

For the second part of the lemma, let $h=I'(i+1,j)$, and thus $d_P(v_{i+1},v_h)\geq |v_{i+1}v_{j}|+d_P(v_j,v_h)$. Note that $d_P(v_i,v_h)-d_P(v_{i+1},v_h)=d_P(v_i,v_{i+1})=|v_iv_{i+1}|$. By the triangle inequality, we have  $|v_iv_j|+d_P(v_j,v_h)-(|v_{i+1}v_j|+d_P(v_j,v_h))=|v_iv_j| - |v_{i+1}v_j|\leq |v_iv_{i+1}|$. Hence, we can derive
$d_P(v_i,v_h) = d_P(v_{i+1},v_h)+|v_iv_{i+1}|\geq (|v_{i+1}v_{j}|+d_P(v_j,v_h)) +|v_iv_{i+1}| \geq |v_iv_j|+d_P(v_j,v_h)$. By the above observation, $I'(i,j)\leq h$.
\qed
\end{proof}

The following lemma characterizes a monotonicity property of the $\beta$ values.

\begin{lemma}\label{lem:40}
	$\beta(i,j)\leq \beta(i,j+1)$ for all $j\in [i,n-1]$.
\end{lemma}
\begin{proof}
Let $h=I(i,j)$. Hence, $\beta(i,j)=d_{C(i,j)}(v_i,v_h)=\min\{d_P(v_i,v_h),|v_iv_j|+d_P(v_j,v_h)\}$.
By the triangle inequality, we have $|v_iv_j|+d_P(v_j,v_h)\leq |v_iv_{j+1}|+d_P(v_{j+1},v_h)$.
Since $\beta(i,j+1)=\max_{k\in [i,j+1]}\min\{d_P(v_i,v_k),|v_iv_{j+1}|+d_P(v_{j+1},v_k)\}$, we obtain
\begin{equation*}
\begin{split}
\beta(i,j+1) & \geq \min\{d_P(v_i,v_h),|v_iv_{j+1}|+d_P(v_{j+1},v_h)\}\\
& \geq \min\{d_P(v_i,v_h),|v_iv_j|+d_P(v_j,v_h)\} =\beta(i,j).
\end{split}
\end{equation*}
%
%
%
%
%
%
%
%
\qed
\end{proof}

Consider any $i\in [1,n]$.
Recall that $j(i)$ is the index $j$ that minimizes the value $\alpha(i,j)$
for all $j\in [i,n]$, and $\alpha(i,j)=\max\{\beta(i,j),\gamma(i,j)\}$. If we consider
$\alpha(i,j)$, $\beta(i,j)$, and $\gamma(i,j)$ as functions of
$j\in [i,n]$, then by Lemmas~\ref{lem:20} and \ref{lem:40},
$\alpha(i,j)$ is a unimodal function (first decreases and then increases; e.g., see Fig.~\ref{fig:alpha}).
In order to compute $j(i)$ and thus $\lambda_i$ during our algorithm,
we define $j'(i)$ to be the smallest index $j\in [i,n]$
such that $\gamma(i,j)\leq \beta(i,j)$. Note that such $j$ must exist because $\gamma(i,n)\leq \beta(i,n)$.
We have the following observation.

\begin{figure}[t]
\begin{minipage}[t]{\textwidth}
\begin{center}
\includegraphics[height=1.3in]{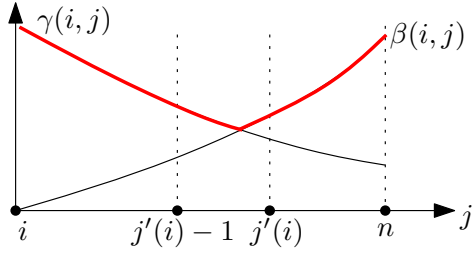}
\caption{\footnotesize Illustrating the two functions $\beta(i,j)$ and $\gamma(i,j)$ for $j\in [i,n]$. The function $\alpha(i,j)$, depicted by the thick (red) curve, is the pointwise maximum of them. The index $j'(i)$ is also shown.}
\label{fig:alpha}
\end{center}
\end{minipage}
\vspace{-0.15in}
\end{figure}

\begin{observation}\label{obser:30}
$j(i)$ is either $j'(i)-1$ or $j'(i)$.
\end{observation}
\begin{proof}
	By the definition of $j'(i)$, $\alpha(i,j'(i))=\beta(i,j'(i))$ and
	$\alpha(i,j'(i)-1)=\gamma(i,j'(i)-1)$.
	Consider any $k\in [i,n]$.
	If $k>j'(i)$, then $\alpha(i,j'(i))=\beta(i,j'(i))\leq
	\beta(i,k)\leq \alpha(i,k)$. If $k<j'(i)-1$, then
	$\alpha(i,j'(i)-1)=\gamma(i,j'(i)-1)\leq
	\gamma(i,k)\leq \alpha(i,k)$.
	Therefore, $\alpha(i,j)$ is minimized at either $j=j'(i)-1$ or $j=j'(i)$.
	\qed
\end{proof}

Our last monotonicity is given in the following lemma, which will help us to determine $j'(i)$.

\begin{lemma}\label{lem:50}
$j'(i)\leq j'(i+1)$ for all $i\in [n-1]$.
\end{lemma}
\begin{proof}
	If $j'(i)\leq i+1$, then the lemma is trivially true as $j'(i+1)\geq i+1$. Otherwise, let
	$h=j'(i)\geq i+2$. According to the definition of $j'(i)$,
	$\gamma(i,h-1)> \beta(i,h-1)$. In the following, we prove that
	$\gamma(i+1,h-1)> \beta(i+1,h-1)$, which would imply that
	$j'(i+1)\geq h=j'(i)$. To this end, since $\gamma(i,h-1)> \beta(i,h-1)$, it is sufficient to prove the following
	\begin{equation*}
	\gamma(i+1,h-1)- \beta(i+1,h-1)\geq \gamma(i,h-1)- \beta(i,h-1).
	\end{equation*}

	Since $\gamma(i,h-1)=|v_{i}v_{h-1}|+d_P(v_{h-1},v_n)$ and
	$\gamma(i+1,h-1)=|v_{i+1}v_{h-1}|+d_P(v_{h-1},v_n)$, the above
	inequality is equivalent to the following
	\begin{equation}\label{equ:10}
		|v_{i+1}v_{h-1}| + \beta(i,h-1)\geq |v_{i}v_{h-1}|+ \beta(i+1,h-1).
	\end{equation}

	Let $k=I(i+1,h-1)$. By definition, we have
	$\beta(i+1,h-1)=d_{C(i+1,h-1)}(v_{i+1},v_k)=\min\{d_P(v_{i+1},v_k),|v_{i+1}v_{h-1}|+d_P(v_{h-1},v_k)\}$.

	Since $k\in [i+1,h-1]$, observe that $\beta(i,h-1)\geq \min\{d_P(v_{i},v_k),|v_{i}v_{h-1}|+d_P(v_{h-1},v_k)\}$.
Further, by the triangle inequality, we have 
$|v_{i+1}v_{h-1}|+ d_P(v_i,v_{k})  = |v_{i+1}v_{h-1}|+ d_P(v_i,v_{i+1}) +
			d_P(v_{i+1},v_k) \geq |v_{i}v_{h-1}|+ d_P(v_{i+1},v_k)$. Hence, we can derive
\begin{equation*}
		\begin{split}
|v_{i+1}v_{h-1}| + \beta(i,h-1) & \geq |v_{i+1}v_{h-1}| + \min\{d_P(v_i,v_k),|v_iv_{h-1}|+d_P(v_{h-1},v_k)\}\\
& = \min\{|v_{i+1}v_{h-1}|+d_P(v_i,v_k),|v_{i+1}v_{h-1}|+ |v_iv_{h-1}|+d_P(v_{h-1},v_k)\}\\
& \geq\min\{|v_{i}v_{h-1}|+ d_P(v_{i+1},v_k), |v_{i+1}v_{h-1}|+ |v_iv_{h-1}|+d_P(v_{h-1},v_k) \}\\
&=|v_{i}v_{h-1}|+\min\{d_P(v_{i+1},v_k), |v_{i+1}v_{h-1}|+ d_P(v_{h-1},v_k)\}\\
&= |v_{i}v_{h-1}|+\beta(i+1,h-1),
	\end{split}
	\end{equation*}
which proves Equation~\eqref{equ:10}.\qed
%
%
%
\end{proof}

Based on the above several monotonicity properties, we present our linear time algorithm for computing $\lambda_i$ and $j(i)$ for all $i\in [1,n]$, as follows.
Recall that we have done preprocessing so that $d_P(v_i,v_j)$ can be computed in $O(1)$ time for any pair $(i,j)$.

Starting with $i=1$ and $j=1$, we increment $j$ from $1$ to $n$. For each $j$, we maintain the four values $\gamma(i,j-1)$, $\gamma(i,j)$, $\beta(i,j-1)$, and $\beta(i,j)$. So $\alpha(i,j-1)$ and $\alpha(i,j)$ can be obtained in $O(1)$ time. Since $\gamma(i,j)=|v_iv_j|+d_P(v_j,v_n)$, $\gamma(i,j)$ can be computed in $O(1)$ time, and the same applies to $\gamma(i,j-1)$. We will explain how to compute the $\beta$ values later. During the increasing of $j$, if the first time we find $\gamma(i,j)\leq \beta(i,j)$, then $j'(i)=j$. By Observation~\ref{obser:30}, $\lambda_i=\min\{\alpha(i,j-1),\alpha(i,j)\}$.

Then, we increase $i$ by one (for differentiation, we use $i+1$ to denote the increased $i$). By Lemma~\ref{lem:50}, to determine $j'(i+1)$, we only need to start $j$ from $j=j'(i)$. Following the similar procedure as above, we increase $j$ and maintain $\gamma(i+1,j-1)$, $\gamma(i+1,j)$, $\beta(i+1,j-1)$, and $\beta(i+1,j)$. Initially when $j=j'(i)$, $\gamma(i+1,j-1)$ and $\gamma(i+1,j)$ can be computed in $O(1)$ time as discussed before; for $\beta(i+1,j-1)$ and $\beta(i+1,j)$, we will show later that they can be computed in $O(1)$ amortized time.
In this way, the total time for computing $\lambda_i$ and $j(i)$ for all $i\in [1,n]$ is $O(n)$.

It remains to describe how to compute the values $\beta(i,j)$.
As discussed before, by Observation~\ref{obser:20}, it is sufficient to determine $I'(i,j)$, after which $\beta(i,j)$ can be computed in $O(1)$ time.

Our algorithm relies on the monotonicity properties of Lemma~\ref{lem:30}. Initially, when $i=j=1$, we let $k=1$. As $j$ increases, we also increase $k$. We can compute both $d_P(v_i,v_k)$ and $|v_iv_j|+d_P(v_k,v_j)$ in constant time for each triple $(i,k,j)$. During the increasing of $k$, if we find $d_P(v_i,v_k)\geq |v_iv_j|+d_P(v_k,v_j)$ for the first time, then $I'(i,j)$ is $k$.
After $j$ is increased, we need to compute $\beta(i,j+1)$, i.e., determine $I'(i,j+1)$ (for differentiation, we use $j+1$ to refer to the increased $j$). To this end, by Lemma~\ref{lem:30}, we have $k$ start from $I'(i,j)$, which is the exactly current value of $k$. Similarly, when $i$ is increased and we need to determine $I'(i+1,j)$, we also have $k$ start from $I'(i,j)$, i.e., the current value of $k$. Therefore, in the entire algorithm, the index $k$ continuously increases from $1$ to $n$.

In summary, in the overall algorithm $i$ and $j$ simultaneously increase from $1$ to $n$ with $i\leq j$. Hence, the number of $\beta(i,j)$ values computed in the entire algorithm is at most $2n$. Further, the procedure for computing all $\beta$ values increases $k$ from $1$ to $n$. Thus, the total time for computing all $\beta$ values in the algorithm is $O(n)$, and the amortized time for computing each $\beta$ value is $O(1)$.

This proves Lemma~\ref{lem:10}.

\subsection{The Configurations and Our Algorithm}

In this section, we present our algorithm for computing an optimal solution. As discussed before, we will consider all possible configurations for the optimal solution and compute a candidate solution for each such configuration.

Recall the definitions of $c^*$, $a^*$, $b^*$, $r^*$, and $\pi^*$ in the beginning of Section~\ref{sec:algo}.
We already discussed one configuration above, i.e., $c^*$ is on $P(v_1,v_{i^*})$. With the help of Lemma~\ref{lem:10}, we gave a linear time algorithm for it.
Another configuration, which is symmetric, is that $c^*$ is on $P(v_{j^*},v_n)$. Correspondingly, we can use an analogous algorithm (e.g., reverse the indices of $P$ and then apply the same algorithm) to compute a candidate solution in linear time. We omit the details. For the reference purpose, we consider the above two configurations as Case 0.

It remains to consider the configuration $c^*\in C(i^*,j^*)\setminus\{v_{i^*},v_{j^*}\}$. Here, $c^*$ is in the interior of either $e(v_{i^*},v_{j^*})$ or $P(v_{i^*},v_{j^*})$.
It is not difficult to see that $v_{a^*}$ is either $v_1$ or a vertex in $P(v_{i^*},v_{j^*})$, i.e., $a^*=1$ or $a^*\in [i^*,j^*]$.
Similarly, $b^*=n$ or $b^*\in [i^*,j^*]$. Depending on whether $a^*=1$, $b^*=n$, or both $a^*$ and $b^*$ are in $[{i^*},{j^*}]$, there are three main cases.

\subsubsection{Case 1. $a^*=1$.}

In this case, depending on whether $b^*=n$ or $b^*\in [i^*,j^*]$, there are two cases.

\subsubsection{Case 1.1. $b^*=n$.}

In this case, if $\pi^*$ does not contain $e(i^*,j^*)$, then $\pi^*$ is the entire path $P$. Correspondingly, we keep a candidate solution with $d_P(v_1,v_n)/2$ as the radius.

In the following, we focus on the case where $\pi^*$ contains $e(i^*,j^*)$. As $c^*\in C(v_{i^*},v_{j^*})\setminus\{v_{i^*},v_{j^*}\}$ and $c^*\in \pi^*$, $c^*$ must be in the interior of $e(v_{i^*},v_{j^*})$ (e.g., see Fig.~\ref{fig:conf11}).

\begin{figure}[t]
\begin{minipage}[t]{\textwidth}
\begin{center}
\includegraphics[height=0.7in]{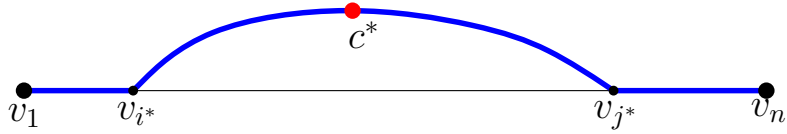}
\caption{\footnotesize Illustrating the configuration for Case 1.1, where $c^*\in e(v_{i^*},v_{j^*})$, $a^*=1$, and $b^*=n$. The thick (blue) path is $\pi^*$.}
\label{fig:conf11}
\end{center}
\end{minipage}
\vspace{-0.15in}
\end{figure}

We make an assumption on $j^*$ that no index $j>j^*$ exists such that $(i^*,j)$ is also an optimal solution with the same  configuration as $(i^*,j^*)$, since otherwise we could instead consider $(i^*,j)$ as $(i^*,j^*)$.
We also assume that none of the previously discussed configurations has an optimal solution since otherwise our previously obtained candidate solutions already have an optimal one.
With these assumptions, we have the following lemma, which is a key observation for our algorithm.

\begin{lemma}\label{lem:60}
Let $k^*$ be the smallest index such that $d_P(v_{i^*},v_{k^*})>d_P(v_1,v_{i^*})$. Such an index $k^*$ must exist in $[i^*,j^*]$. Further, $j^*$ is the largest index $j\in [k^*,n]$ such that $d_P(v_{k^*},v_j)\leq d_P(v_j,v_n)$.
\end{lemma}
\begin{proof}
Assume to the contrary that such an index $k^*$ as stated in the lemma doest not exist in $[i^*,j^*]$. Then, $d_P(v_{i^*},v_j)\leq d_P(v_1,v_{i^*})$ for all $j\in [i^*,j^*]$. Depending on whether $j^*=n$, there are two cases.

If $j^*=n$, then let $c$ be a point on $e(v_{i^*},v_{j^*})$ arbitrarily close to $c^*$ and to the left of $c^*$. Since $c^*$ is in the interior of $e(v_{i^*},v_{j^*})$, such a point $c$ must exist on $e(v_{i^*},v_{j^*})$.
Note that $|c^*v_{i^*}|$ is the length of the sub-edge of $e(v_{i^*},v_{j^*})$ from $c^*$ to $v_{i^*}$. Similarly, $|cv_{i^*}|$ the length of the sub-edge of $e(v_{i^*},v_{j^*})$ from $c$ to $v_{i^*}$.
Since $|cv_{i^*}|<|c^*v_{i^*}|$ and $|c^*v_{i^*}|+d_P(v_{i^*},v_1)=r^*$, $|cv_{i^*}|+d_P(v_{i^*},v_1)<r^*$. On the other hand, for each $j\in [i^*,j^*]$, since $d_P(v_{i^*},v_j)\leq d_P(v_1,v_{i^*})$, we also have $|cv_{i^*}|+d_P(v_{i^*},v_j)<r^*$. Because $j^*=n$, the above implies that the distance from $c$ to all vertices in $G(i^*,j^*)$ is strictly smaller than $r^*$, which contradicts with the definition of $r^*$.

If $j^*<n$, then consider the graph $G(i^*,j^*+1)$ (which is $P\cup\{e(v_{i^*},v_{j^*+1})\}$). Let $c$ be a point on $e(v_{i^*},v_{j^*+1})\cup e(v_{j^*+1},v_{j^*})$ whose distance along $e(v_{i^*},v_{j^*+1})\cup e(v_{j^*+1},v_{j^*})$ from $v_{i^*}$ is $|c^*v_{i^*}|$. Note that such a point $c$ must exist since $|v_{i^*}v_{j^*+1}| + |v_{j^*+1}v_{j^*}|\geq |v_{i^*}v_{j^*}|$ due to the triangle inequality. The point $c$ is either on $e(v_{i^*},v_{j^*+1})\setminus\{v_{j^*+1}\}$ or on $e(v_{j^*},v_{j^*+1})$. We analyze the two cases below.

\begin{figure}[t]
\begin{minipage}[t]{0.49\textwidth}
\begin{center}
\includegraphics[height=0.6in]{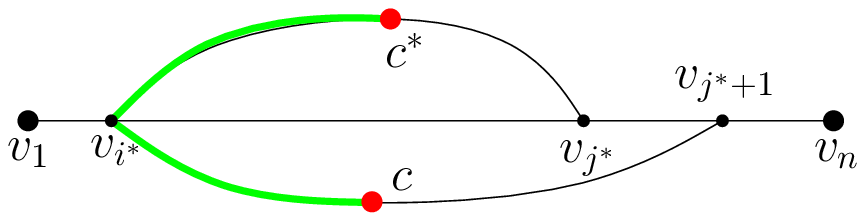}
\caption{\footnotesize Illustrating the proof of Lemma~\ref{lem:60} where $c\in e(v_{i^*},v_{j^*+1})\setminus\{v_{j^*+1}\}$: The two thick (green) curves have the same length.}
\label{fig:newgraph}
\end{center}
\end{minipage}
\hspace{0.05in}
\begin{minipage}[t]{0.49\textwidth}
\begin{center}
\includegraphics[height=0.6in]{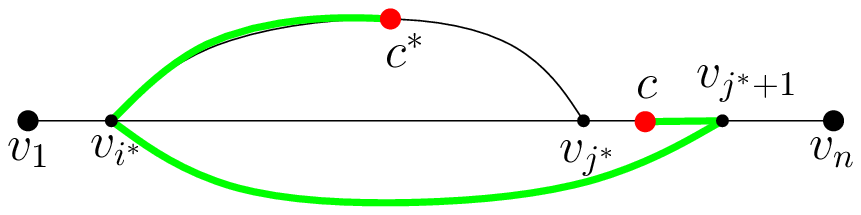}
\caption{\footnotesize Illustrating the proof of Lemma~\ref{lem:60} where $c\in e(v_{j^*},v_{j^*+1})$: The two thick (green) curves have the same length.}
\label{fig:newgraph10}
\end{center}
\end{minipage}
\vspace{-0.15in}
\end{figure}

If $c$ is on $e(v_{i^*},v_{j^*+1})\setminus\{v_{j^*+1}\}$ (e.g., see Fig.~\ref{fig:newgraph}), then we claim that the distances from $c$ to all vertices in $G(i^*,j^*+1)$ are no more than $r^*$. Indeed, for each $j\in [1,i^*]$, by the definition of $c$, $d_{G(i^*,j^*+1)}(c,v_j)\leq |cv_{i^*}|+d_P(v_{i^*},v_j)\leq |cv_{i^*}|+d_P(v_{i^*},v_1)=|c^*v_{i^*}|+d_P(v_{i^*},v_1)=r^*$. For each $j\in [i^*,j^*]$, since $d_P(v_{i^*},v_j)\leq d_P(v_{i^*},v_1)$, we have $d_{G(i^*,j^*+1)}(c,v_j)\leq |cv_{i^*}|+d_P(v_{i^*},v_j)\leq |cv_{i^*}|+d_P(v_{i^*},v_1)= r^*$.
For each $j\in [j^*+1,n]$, by the triangle inequality, we can derive
\begin{equation}\label{equ:30}
\begin{split}
d_{G(i^*,j^*+1)}(c,v_j)& \leq |cv_{j^*+1}|+d_P(v_{j^*+1},v_j)
 =|v_{i^*}v_{j^*+1}|-|v_{i^*}c|+d_P(v_{j^*+1},v_j)\\
& \leq |v_{i^*}v_{j^*}|+|v_{j^*}v_{j^*+1}|-|v_{i^*}c|+d_P(v_{j^*+1},v_j)\\
& =|v_{i^*}v_{j^*}|+|v_{j^*}v_{j^*+1}|-|v_{i^*}c^*|+d_P(v_{j^*+1},v_j)\\
& =|c^*v_{j^*}|+|v_{j^*}v_{j^*+1}|+d_P(v_{j^*+1},v_j) \\
& =|c^*v_{j^*}|+d_P(v_{j^*},v_j) \leq |c^*v_{j^*}|+d_P(v_{j^*},v_n) = r^*.
\end{split}
\end{equation}
The above proves the claim, which implies that $(i^*,j^*+1)$ is also an optimal solution with the same configuration as $(i^*,j^*)$. But this contradicts with our previous assumption on $j^*$: no index $j>j^*$ exists such that $(i^*,j)$ is also an optimal solution with the same  configuration as $(i^*,j^*)$.

If $c$ is on $e(v_{j^*},v_{j^*+1})$ (e.g., see Fig.~\ref{fig:newgraph10}), then let $c'=v_{j^*+1}$. We claim that the distances from $c'$ to all vertices in $G(i^*,j^*+1)$ are at most $r^*$. We prove the claim below.

On the one hand, for any $j\in [1,j^*]$, by the definition of $c'$, we have
\begin{equation*}
\begin{split}
d_{G(i^*,j^*+1)}(c',v_j) & \leq |c'v_{i^*}|+d_P(v_{i^*},v_j) \leq |c'v_{i^*}|+d_P(v_{i^*},v_1) \\
& \leq |cv_{j^*+1}|+|v_{j^*+1}v_{i^*}|+d_P(v_{i^*},v_1) \\
&= |c^*v_{i^*}|+d_P(v_{i^*},v_1) = r^*.\\
\end{split}
\end{equation*}
Note that because $c\in e(v_{j^*},v_{j^*+1})$, $|cv_{j^*+1}|$ is the length of the sub-edge of $e(v_{j^*},v_{j^*+1})$ from $c$ to $v_{j^*+1}$.

On the other hand, for any $j\in [j^*+1,n]$, by the definition of $c'$, we have
\begin{equation*}
\begin{split}
d_{G(i^*,j^*+1)}(c',v_j)\leq d_P(c',v_{j}) & =d_P(v_{j^*+1},v_{j})\leq d_P(v_{j^*+1},v_{n}) \leq |c^*v_{j^*}|+d_P(v_{j^*},v_n) = r^*.
\end{split}
\end{equation*}

The above proves the claim. The claim implies that the graph $G(i^*,j^*+1)$ has a center in $P(v_{j^*+1},v_n)$, which is a configuration of Case 0. But this contradicts with our assumption that none of the previously discussed configurations has an optimal solution.

The above proves the index $k^*$ must exit in $[i^*,j^*]$.
In the following, we prove that $j^*$ is the largest index $j\in [k^*,n]$ with $d_P(v_{k^*},v_j)\leq d_P(v_j,v_n)$. The proof techniques are somewhat similar. Note that there must exist an index $j\in[k^*,n]$ such that $d_P(v_{k^*},v_j)\leq d_P(v_j,v_n)$ (e.g., $j=k^*$).
If $j^*=n$, then the statement is obviously true. In the following we assume that $j^*<n$.

Assume to the contrary that $j^*$ is not the largest such index. Then, it must be true that $d_P(v_{k^*},v_{j^*+1})\leq d_P(v_{j^*+1},v_n)$. Consider the graph $G(i^*,j^*+1)$. As above, let $c$ be a point on $e(v_{i^*},v_{j^*+1})\cup e(v_{j^*+1},v_{j^*})$ whose distance along $e(v_{i^*},v_{j^*+1})\cup e(v_{j^*+1},v_{j^*})$ from $v_{i^*}$ is $|c^*v_{i^*}|$.

If $c$ is on $e(v_{i^*},v_{j^*+1})\setminus\{v_{j^*+1}\}$, then we claim that the distances from $c$ to all vertices in $G(i^*,j^*+1)$ are no more than $r^*$. For each $j\in [1,k^*-1]$, by the same analysis as above, $d_{G(i^*,j^*+1)}(c,v_j)\leq |cv_{i^*}|+d_P(v_{i^*},v_j)\leq |cv_{i^*}|+d_P(v_{i^*},v_1)=r^*$.
For each $j\in [k^*,n]$, $d_{G(i^*,j^*+1)}(c,v_j)\leq |cv_{j^*+1}|+d_P(v_{j^*+1},v_j)\leq |cv_{j^*+1}|+d_P(v_{j^*+1},v_n)$.
By Equation~\eqref{equ:30}, we can show that $|cv_{j^*+1}|+d_P(v_{j^*+1},v_n)\leq r^*$. The claim is thus proved. The claim implies that $(i^*,j^*+1)$ is also an optimal solution with the same configuration as $(i^*,j^*)$. But this contradicts with our assumption on $j^*$.

If $c$ is on $e(v_{j^*},v_{j^*+1})$, then let $c'=v_{j^*+1}$. By the same analysis as above, we can show that the distances from $c'$ to all vertices in $G(i^*,j^*+1)$ are no more than $r^*$. Since $c'\in P(v_{j^*+1},v_n)$, the claim implies that there is an optimal solution conforming with the configuration of Case 0. This again incurs contradiction.\qed
\end{proof}

Based on Lemma~\ref{lem:60}, our algorithm for this case works as follows.

For each index $i\in [1,n]$, define $k(i)$ as the smallest index $k\in [i,n]$ such that $d_P(v_i,v_k)>d_P(v_1,v_i)$ (let $k(i)=n+1$ if no such index exists), and if $k(i)\leq n$, define $j(i)$\footnote{This notation was used before for a different meaning. Because we have many cases to consider, to save notation, we will repeatedly use the same notation as long as the context is clear.} as the largest index $j\in [k(i),n]$ such that $d_P(v_{k(i)},v_j)\leq d_P(v_j,v_n)$ (let $j(i)=n+1$ if $k(i)=n+1$).
The following observation is self-evident.

\begin{observation}
For each $i\in [1,n-1]$, $k(i)\leq k(i+1)$ and $j(i)\leq j(i+1)$.
\end{observation}

By the above observation, we can easily compute $k(i)$ and $j(i)$ for
all $i\in [1,n]$ in $O(n)$ time by a linear scan on $P$. We omit the
details.

For each $i$, if $j(i)\leq n$, let
$r(i)=(d_P(v_1,v_i)+|v_iv_{j(i)}|+d_P(v_{j(i)},v_n))/2$, and if
$d_P(v_1,v_i)<r(i)$ and $d_P(v_{j(i)},v_n)<r(i)$ (implies that the
center is in the interior of $e(v_i,v_{j(i)})$), then we keep
$(i,j(i))$ as a candidate solution with $r(i)$ as the radius (and the center is a point $c$ in $e(v_i,v_{j(i)})$ with $d_P(v_1,v_i)+|v_ic|=r(i)$). Note
that due to our definitions of $k(i)$ and $j(i)$, the
solution is feasible, i.e., the distances
from $c$ to all vertices in the graph $G(i,j(i))$ are no more than $r(i)$.
The above computes at most $n$ candidate solutions, and among them, we keep the one with the smallest $r(i)$ value as our candidate solution for this configuration. Based on our discussions, if this
configuration has an optimal solution, then our solution is also optimal.
The running time of the algorithm is $O(n)$.

\subsubsection{Case 1.2: $b^*\in [i^*,j^*]$.}

Note that $\pi^*$ either contains $e(i^*,j^*)$ or does not
contain any interior point of the edge.
Depending on whether $\pi^*$ contains $e(i^*,j^*)$, there are two cases.

\subsubsection{Case 1.2.1: $\pi^*$ contains $e(i^*,j^*)$.}

Recall that $c^*$ is in the interior of either $e(v_{i^*},v_{j^*})$ or $P(v_{i^*},v_{j^*})$. We discuss the two cases below.

\subsubsection{Case 1.2.1.1: $c^*\in e(i^*,j^*)$, e.g., see Fig.~\ref{fig:config20}.}

\begin{figure}[t]
\begin{minipage}[t]{\textwidth}
\begin{center}
\includegraphics[height=0.7in]{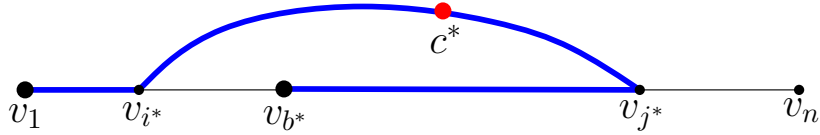}
\caption{\footnotesize  Illustrating the configuration for Case 1.2.1.1, where $c^*\in e(v_{i^*},v_{j^*})$, $a^*=1$, and $b^*\in [i^*,j^*]$. The thick (blue) path is $\pi^*$..}
\label{fig:config20}
\end{center}
\end{minipage}
\vspace{-0.15in}
\end{figure}

Our algorithm for this case is somewhat similar to that for Case 1.1.
We make an assumption on $j^*$ that no index $j<j^*$ exists such that $(i^*,j)$ is also an optimal solution with the same  configuration as $(i^*,j^*)$ since otherwise we could instead consider $(i^*,j)$ as $(i^*,j^*)$. We also assume that none of the previously discussed configurations has an optimal solution. We have the following lemma.

\begin{lemma}\label{lem:70}
\begin{enumerate}
\item
Let $k^*$ be the smallest index such that $d_P(v_{i^*},v_{k^*})>d_P(v_1,v_{i^*})$. Such an index $k^*$ must exist in $[i^*,j^*]$.
\item
$b^*=k^*$.
\item $j^*$ is the smallest index $j\in [k^*,n]$ such that $d_P(v_{k^*},v_j)> d_P(v_j,v_n)$.
\end{enumerate}
\end{lemma}
\begin{proof}
We first prove that $k^*$ exists. Assume to the contrary that this is not true. Then, $d_P(v_{i^*},v_{j})\leq d_P(v_1,v_{i^*})$ for all $j\in [i^*,j^*]$.
Let $c$ be a point of $e(i^*,j^*)$ to the left of $c^*$ and arbitrarily close to $c^*$. We show that the distances from $c$ to all vertices in $G(i^*,j^*)$ are all smaller than $r^*$, incurring contradiction.

Because none of the previously discussed configurations has an optimal solution, $v_n$ is not a farthest point of $c^*$, i.e., $|c^*v_{j^*}|+d_P(v_{j^*},v_n)<r^*$ (since otherwise Case 1.1 would happen). By the definition of $c$, $|cv_{j^*}|+d_P(v_{j^*},v_n)<r^*$.
Hence, for each $j\in [j^*,n]$, $d_{G(i^*,j^*)}(c,v_j)\leq |cv_{j^*}|+d_P(v_{j^*},v_j)\leq |cv_{j^*}|+d_P(v_{j^*},v_n)<r^*$. On the other hand, for each $j\in[1,j^*]$, $d_{G(i^*,j^*)}(c,v_j)\leq |cv_{i^*}|+d_P(v_{i^*},v_j)\leq |cv_{i^*}|+d_P(v_{i^*},v_1)<|c^*v_{i^*}|+d_P(v_{i^*},v_1)=r^*$.

This proves that $k^*$ must exist.

For the second part of the lemma,
notice that the configuration tells that the shortest path from $c^*$ to $v_{b^*}$ in $G(i^*,j^*)$ is the union of the sub-edge of $e(i^*,j^*)$ from $c^*$ to $j^*$ and $P(v_{b^*},v_{j^*})$. Hence, $d_{G(i^*,j^*)}(c^*,v_{b^*})=|c^*v_{j^*}|+d_P(v_{j^*},v_{b^*})\leq |c^*v_{i^*}|+d_P(v_{i^*},v_{b^*})$.
Note that $|c^*v_{j^*}|+d_P(v_{j^*},v_{b^*})= |c^*v_{i^*}|+d_P(v_{i^*},v_{b^*})$ is not possible. Indeed, if this were true, then since $|c^*v_{j^*}|+d_P(v_{j^*},v_n)<r^*$, if we move $c^*$ slightly on $e(v_{i^*},v_{j^*})$ towards $v_{i^*}$, then the distances from $c^*$ to all vertices in $G(i^*,j^*)$ would be all smaller than $r^*$, incurring contradiction. Therefore, $|c^*v_{j^*}|+d_P(v_{j^*},v_{b^*})< |c^*v_{i^*}|+d_P(v_{i^*},v_{b^*})$.
Since $|c^*v_{j^*}|+d_P(v_{j^*},v_{b^*})=|c^*v_{i^*}|+d_P(v_{i^*},v_{1})=r^*$, we obtain $d_P(v_{i^*},v_{1})< d_P(v_{i^*},v_{b^*})$.
The definition of $k^*$ implies that $k^*\leq b^*$.

On the other hand, assume to the contrary that $k^*<b^*$.
Then, $|c^*v_{j^*}|+d_P(v_{j^*},v_{k^*})>|c^*v_{j^*}|+d_P(v_{j^*},v_{b^*})=d_{G(i^*,j^*)}(c^*,v_{b^*})=r^*$.
Further, since $d_P(v_{i^*},v_{k^*})>d_P(v_1,v_{i^*})$, $|c^*v_{i^*}|+d_P(v_{i^*},v_{k^*})>|c^*v_{i^*}|+d_P(v_{i^*},v_{1})=r^*$.
Consequently, $d_{G(i^*,j^*)}(c^*,v_{k^*})=\min\{|c^*v_{j^*}|+d_P(v_{j^*},v_{k^*}),|c^*v_{i^*}|+d_P(v_{i^*},v_{k^*})\}>r^*$, a contradiction. This proves that $k^*= b^*$.

We proceed to prove the third part of the lemma. Since $|c^*v_{j^*}|+d_P(v_{j^*},v_n)<r^*$  and $|c^*v_{j^*}|+d_P(v_{j^*},v_{k^*})=r^*$, we have $d_P(v_{k^*},v_{j^*})>d_P(v_{j^*},v_n)$.
Assume to the contrary the statement is not true. Then, it must hold that $d_P(v_{k^*},v_{j^*-1})> d_P(v_{j^*-1},v_n)$.

Let $c$ be a point on $e(v_{i^*},v_{j^*-1})\cup e(v_{j^*-1},v_{j^*})$ whose distance along $e(v_{i^*},v_{j^*-1})\cup e(v_{j^*-1},v_{j^*})$ from $v_{i^*}$ is $|c^*v_{i^*}|$. Such a point $c$ must exist since $|v_{i^*}v_{j^*-1}| + |v_{j^*-1}v_{j^*}|\geq |v_{i^*}v_{j^*}|$. Depending on whether $c$ is on $e(v_{i^*},v_{j^*-1})\setminus\{v_{j^*-1}\}$ or on $e(v_{j^*-1},v_{j^*})$, there are two cases. The proof is similar to that for Lemma~\ref{lem:60}. So we briefly discuss it.

\begin{figure}[t]
\begin{minipage}[t]{0.49\textwidth}
\begin{center}
\includegraphics[height=0.6in]{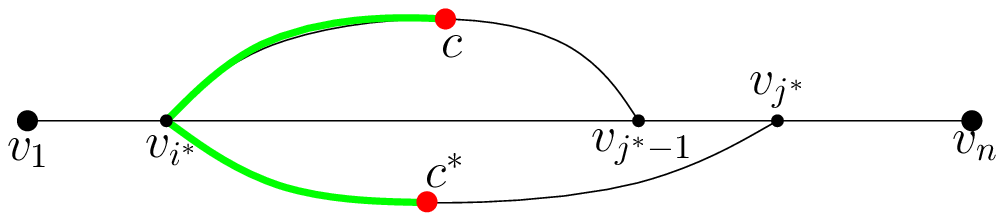}
\caption{\footnotesize Illustrating the proof of Lemma~\ref{lem:70} where $c\in e(v_{i^*},v_{j^*-1})\setminus\{v_{j^*-1}\}$: The two thick (green) curves have the same length.}
\label{fig:newgraph20}
\end{center}
\end{minipage}
\hspace{0.05in}
\begin{minipage}[t]{0.49\textwidth}
\begin{center}
\includegraphics[height=0.6in]{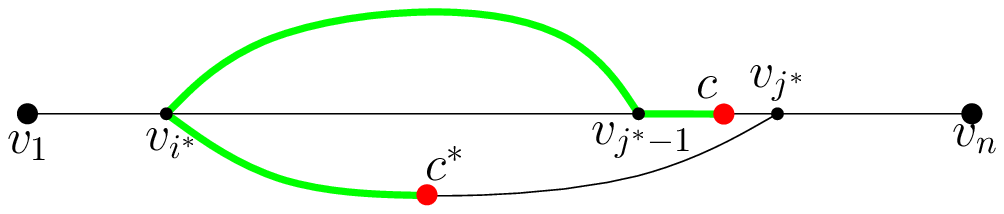}
\caption{\footnotesize Illustrating the proof of Lemma~\ref{lem:70} where $c\in e(v_{j^*-1},v_{j^*})$: The two thick (green) curves have the same length.}
\label{fig:newgraph30}
\end{center}
\end{minipage}
\vspace{-0.15in}
\end{figure}

If $c$ is on $e(v_{i^*},v_{j^*-1})\setminus\{v_{j^*-1}\}$, e.g., see Fig.~\ref{fig:newgraph20}, then we claim that the distances from $c$ to all vertices in $G(i^*,j^*-1)$ are at most $r^*$. Indeed, for each $j\in [1,k^*-1]$, $d_{G(i^*,j^*-1)}(c,v_j)\leq |cv_{i^*}|+d_P(v_{i^*},v_{j})=|c^*v_{i^*}|+d_P(v_{i^*},v_{j})\leq |c^*v_{i^*}|+d_P(v_{i^*},v_{1})=r^*$.
For each $j\in [k^*,n]$, since $d_P(v_{k^*},v_{j^*-1})> d_P(v_{j^*-1},v_n)$, $d_{G(i^*,j^*-1)}(c,v_j)\leq |cv_{j^*-1}|+d_P(v_{j^*-1},v_{j})\leq |cv_{j^*-1}|+d_P(v_{j^*-1},v_{k^*})$. In a similar way as Equation~\eqref{equ:30}, we can show that $|cv_{j^*-1}|+d_P(v_{j^*-1},v_{k^*})\leq |c^*v_{j^*}|+d_P(v_{j^*},v_{k^*})=r^*$. The claim thus follows.
Hence, $(i^*,j^*-1)$ is also an optimal solution with the same configuration as $(i^*,j^*)$. But this contradicts with our assumption on $j^*$: no index $j<j^*$ exists such that $(i^*,j)$ is also an optimal solution with the same  configuration as $(i^*,j^*)$.

If $c$ is on $e(v_{j^*},v_{j^*-1})$, e.g., see Fig.~\ref{fig:newgraph30}, then let $c'=v_{j^*-1}$. We claim that the distances from $c'$ to all vertices in $G(i^*,j^*-1)$ are at most $r^*$.
On the one hand, for any $j\in [1,k^*-1]$, $d_{G(i^*,j^*-1)}(c',v_j)\leq |c'v_{i^*}|+d_P(v_{i^*},v_j)\leq |c'v_{i^*}|+d_P(v_{i^*},v_1)\leq  |cc'|+|c'v_{i^*}|+d_P(v_{i^*},v_1)=r^*$. On the other hand, for any $j\in [k^*,n]$, since $d_P(v_{k^*},v_{j^*-1})> d_P(v_{j^*-1},v_n)$, $d_{G(i^*,j^*-1)}(c,v_j)\leq d_P(c',v_j)\leq d_P(c',v_{k^*})\leq d_P(v_{j^*},v_{k^*})\leq |c^*v_{j^*}|+d_P(v_{j^*},v_{k^*})=r^*$. This proves the claim.
The claim implies that $G(i^*,j^*-1)$ is also an optimal solution with a center at $v_{j^*-1}\in P(v_{j^*-1},v_n)$. But this is a configuration in Case~0, which contracts with our assumption that none of the previously discussed configurations has an optimal solution.  \qed
\end{proof}

Based on Lemma~\ref{lem:70}, our algorithm for this case works as follows.

For each index $i\in [1,n]$, define $k(i)$ as the smallest $k\in [i,n]$ such
that $d_P(v_i,v_k)>d_P(v_1,v_i)$ (let $k(i)=n+1$ if no such index exists),
and if $k(i)\leq n$, define $j(i)$ as the smallest
index $j\in [k(i),n]$ such that $d_P(v_{k(i)},v_j)> d_P(v_j,v_n)$ (let
$j(i)=n+1$ if no such index exists or if $k(i)=n+1$). It is not
difficult to see that for each $i\in [1,n-1]$, $k(i)\leq k(i+1)$ and
$j(i)\leq j(i+1)$.
The indices $k(i)$ and $j(i)$ can be computed in $O(n)$ time by a linear scan on $P$.
We omit the details.

For each $i$, if $j(i)\leq n$, then let
$r(i)=(d_P(v_1,v_i)+|v_iv_{j(i)}|+d_P(v_{j(i)},v_{k(i)}))/2$, and if
$d_P(v_1,v_i)<r(i)$ and $d_P(v_{j(i)},v_{k(i)})<r(i)$ (this implies that
the center is in the interior of $e(v_i,v_{j(i)})$), then we have a
candidate solution $(i,j(i))$ with $r(i)$ as the radius. By
our definitions of $k(i)$ and $j(i)$, the solution is feasible.
Finally, among the at most $n$ candidate solutions, we
keep the one with the smallest $r(i)$ as our solution for this
configuration. The running time of the algorithm is $O(n)$.

\subsubsection{Case 1.2.1.2: $c^*\in P(i^*,j^*)$.}

Since $\pi^*$ contains $e(v_{i^*},v_{j^*})$, $c^*$ must be to the right of $v_{b^*}$
(e.g., see the bottom example in Fig.~\ref{fig:config}).
Further, $d_P(v_{1},v_{c^*})=d_P(v_1,v_{i^*})+|v_{i^*}v_{j^*}|+d_P(c^*,v_{j^*})=r^*$.
We make an assumption on $j^*$ that no index $j<j^*$ exists such that $(i^*,j)$ is also an optimal solution with the same configuration as $(i^*,j^*)$. We also assume that none of the previously discussed configurations has an optimal solution. The following lemma is literally the same as Lemma~\ref{lem:70}, although the proof is different.

\begin{lemma}\label{lem:75}
\begin{enumerate}
\item
Let $k^*$ be the smallest index such that $d_P(v_{i^*},v_{k^*})>d_P(v_1,v_{i^*})$. Such an index $k^*$ must exist in $[i^*,j^*]$.
\item
$b^*=k^*$.
\item $j^*$ is the smallest index $j\in [k^*,n]$ such that $d_P(v_{k^*},v_j)> d_P(v_j,v_n)$.
\end{enumerate}
\end{lemma}
\begin{proof}
We first prove that $k^*$ exists in $[i^*,j^*]$. Define $j'$ to be the largest index in $[i^*,j^*]$ such that $v_{j'}$ is to the left of or at $c^*$. Since $j'\leq j^*$, it is sufficient to prove that such an index $k^*$ as stated in the lemma must exist in $[i^*,j']$. Assume to the contrary that this is
	not true. Then, $d_P(v_{i^*},v_{j})\leq d_P(v_1,v_{i^*})$ for all
	$j\in [i^*,j']$. Let $c$ be a point of $P(i^*,j^*)$ to the right
	of $c^*$ and arbitrarily close to $c^*$. As $c^*$ is in the interior of $P(i^*,j^*)$, such a point $c$ must exist. In the following, we
	claim that the distances from $c$ to all vertices in $G(i^*,j^*)$ are smaller than
	$r^*$, incurring contradiction.

Note that $P(c^*,v_n)$ is a shortest path from $c^*$ to $v_n$ in $G(i^*,j^*)$. Hence, $d_P(c^*,v_{n})\leq r^*$. By the definition of $c$, for each $j\in [j'+1,n]$, $d_{G(i^*,j^*)}(c,v_j)\leq d_P(c,v_j)\leq d_P(c,v_n)<d_P(c^*,v_n)\leq r^*$. On
	the other hand, for each $j\in[1,j']$, $d_{G(i^*,j^*)}(c,v_j)\leq d_P(c,v_{j^*})+|v_{j^*}v_{i^*}|+d_P(v_{i^*},v_j)\leq
	d_P(c,v_{j^*})+|v_{j^*}v_{i^*}|+d_P(v_{i^*},v_1)<d_P(c^*,v_{j^*})+|v_{j^*}v_{i^*}|+d_P(v_{i^*},v_1)=r^*$.
	The above claim is thus proved.

For the second part of the lemma, the proof is similar to that in Lemma~\ref{lem:70}.
Recall that
$d_{G(i^*,j^*)}(c^*,v_{b^*})=d_P(c^*,v_{b^*})=r^*\leq d_P(c^*,v_{j^*})+|v_{i^*}v_{j^*}|+d_P(v_{i^*},v_{b^*})$.
Note that $r^*=d_P(c^*,v_{j^*})+|v_{i^*}v_{j^*}|+d_P(v_{i^*},v_{b^*})$ is not possible. Indeed, if this were true, then if we move $c^*$ slightly on $P(v_{i^*},v_{j^*})$ towards $v_{j^*}$, the distances from $c^*$ to all vertices in $G(i^*,j^*)$ would be all smaller than $r^*$, incurring contradiction. Therefore, $d_P(c^*,v_{b^*})< d_P(c^*,v_{j^*})+|v_{j^*}v_{i^*}|+d_P(v_{i^*},v_{b^*})$.
Since $d_P(v_{b^*},{c^*})=d_P(v_1,v_{i^*})+|v_{i^*}v_{j^*}|+d_P(v_{j^*},c^*)=r^*$, we obtain $d_P(v_1,v_{i^*})+|v_{i^*}v_{j^*}|+d_P(v_{j^*},c^*)< d_P(c^*,v_{j^*})+|v_{j^*}v_{i^*}|+d_P(v_{i^*},v_{b^*})$, and thus $d_P(v_1,v_{i^*})<d_P(v_{i^*},v_{b^*})$.
The definition of $k^*$ implies that $k^*\leq b^*$.

On the other hand, assume to the contrary that $k^*<b^*$.
Then, $d_P(c^*,v_{k^*})>d_P(c^*,v_{b^*})=d_{G(i^*,j^*)}(c^*,v_{b^*})=r^*$.
Further, since $d_P(v_{i^*},v_{k^*})>d_P(v_1,v_{i^*})$, we have $d_P(c^*,v_{j^*})+|v_{j^*}v_{i^*}|+d_P(v_{i^*},v_{k^*})>d_P(c^*,v_{j^*})+|v_{j^*}v_{i^*}|+d_P(v_{i^*},v_{1})=r^*$. Thus, we obtain $d_{G(i^*,j^*)}(c^*,v_{k^*})=\min\{d_P(c^*,v_{k^*}),d_P(c^*,v_{j^*})+|v_{j^*}v_{i^*}|+d_P(v_{i^*},v_{k^*})\}>r^*$, a contradiction. This proves $k^*= b^*$.

We proceed to prove the third part of the lemma.

	We first prove that $c^*$ must be in the interior of $e(v_{j^*-1},v_{j^*})$.
	Assume to the contrary that this is not true. Then, since $c^*$ is not at $v_{j^*}$, $c^*$ is in
	$P(v_{k^*},v_{j^*-1})$.

We claim that the distances from $c^*$ to
	all vertices in $G(i^*,j^*-1)$ are at most $r^*$.
	Indeed, for each $j\in [k^*,n]$, $d_{G(i^*,j^*-1)}(c^*,v_j)\leq d_P(c^*,v_j)\leq d_P(c^*,v_{k^*})=r^*$. For each $j\in [1,k^*-1]$,
	$d_{G(i^*,j^*-1)}(c^*,v_j)\leq d_P(c^*,v_{j^*-1})+|v_{j^*-1}v_{i^*}|+d_P(v_{i^*},v_j)\leq
	d_P(c^*,v_{j^*-1})+|v_{j^*-1}v_{i^*}|+d_P(v_{i^*},v_{1})\leq
	d_P(c^*,v_{j^*-1})+d_P(v_{j^*-1},v_{j^*})+|v_{j^*}v_{i^*}|+d_P(v_{i^*},v_{1})=r^*$.
	The last inequality is due to the triangle inequality. This proves the claim.
    The claim implies that $(i^*,j^*-1)$ is also an optimal solution with center at $c^*$. If $c^*$ is at $v_{j^*-1}\in P(v_{j^*-1},v_n)$, then this is a configuration of Case 0, which contradicts with our assumption that none of the previously discussed configurations has an optimal solution. Otherwise, the optimal solution $(i^*,j^*-1)$ has the same configuration as $(i^*,j^*)$, which contradicts with our assumption on $j^*$: no index $j<j^*$ exists such that $(i^*,j)$ is also an optimal solution with the same configuration as $(i^*,j^*)$.

	The above proves that $c^*$ is in the interior of $e(v_{j^*-1},v_{j^*})$.

We claim that $d_P(v_{k^*},v_{j^*})> d_P(v_{j^*},v_n)$. Indeed, since $c^*$ is in the interior of $e(v_{j^*-1},v_{j^*})$, we have $d_P(v_{k^*},v_{j^*}) > d_P(v_{k^*},c^*) = r^* \geq d_P(c^*,v_n)> d_P(v_{j^*},v_n)$.
For the third part of the lemma, assume to the contrary the statement is not true.
	Then, due to the above claim, it must hold that $d_P(v_{k^*},v_{j^*-1})> d_P(v_{j^*-1},v_n)$.
	Depending on whether $|v_{i^*}v_{j^*-1}|<|c^*v_{j^*}|+|v_{i^*}v_{j^*}|$, there are two cases.

	If $|v_{i^*}v_{j^*-1}|<|c^*v_{j^*}|+|v_{i^*}v_{j^*}|$,
	then let $c=v_{j^*-1}$. In
	the following, we show that the distances from $c$ to all vertices
	in $G(i^*,j^*-1)$ are smaller than $r^*$, which would incur contradiction.
	Indeed, since $c^*$ is in the interior of $e(v_{j^*-1},v_{j^*})$, $c$ is strictly to the
	left of $c^*$. Thus, for each $j\in
	[k^*,j^*-1]$, $d_{G(i^*,j^*-1)}(c,v_j)\leq d_P(v_{j},c)\leq d_P(v_{k^*},c)<
	d_P(v_{k^*},c^*)=r^*$. For each $j\in [j^*,n]$, since $d_P(v_{k^*},v_{j^*-1})> d_P(v_{j^*-1},v_n)$, $d_{G(i^*,j^*-1)}(c,v_j)\leq d_P(c,v_j)\leq d_P(c,v_n)=d_P(v_{j^*-1},v_n)<d_P(v_{k^*},v_{j^*-1})=d_P(v_{k^*},c)<r^*$.
	For each $j\in [1,k^*-1]$, since
	$|v_{i^*}v_{j^*-1}|<|c^*v_{j^*}|+|v_{i^*}v_{j^*}|$, $d_{G(i^*,j^*-1)}(c,v_j)\leq |v_{j^*-1}v_{i^*}|+d_P(v_{i^*},v_j)\leq
	|v_{j^*-1}v_{i^*}|+d_P(v_{i^*},v_1)<|c^*v_{j^*}|+|v_{j^*}v_{i^*}|+d_P(v_{i^*},v_1)=r^*$.
	The claim thus follows.

\begin{figure}[t]
\begin{minipage}[t]{\textwidth}
\begin{center}
\includegraphics[height=0.6in]{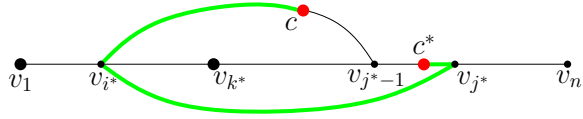}
\caption{\footnotesize Illustrating the proof of Lemma~\ref{lem:75} where $c\in e(v_{j^*-1},v_{j^*})$: The two thick (green) curves from $v_{i^*}$ have the same length.}
\label{fig:newgraph40}
\end{center}
\end{minipage}
\vspace{-0.15in}
\end{figure}

	If $|v_{i^*}v_{j^*-1}|\geq |c^*v_{j^*}|+|v_{i^*}v_{j^*}|$, then
	let $c$ be the point on $e(v_{i^*},v_{j^*-1})$ such that
	$|cv_{i^*}|=|c^*v_{j^*}|+|v_{j^*}v_{i^*}|$, e.g., see Fig.~\ref{fig:newgraph40}. We claim that the
	distances from $c$ to all vertices in $G(i^*,j^*-1)$ are at most $r^*$.
	Indeed, for each $j\in [1,k^*-1]$, $d_{G(i^*,j^*-1)}(c,v_j)\leq
	|cv_{i^*}|+d_P(v_{i^*},v_j)\leq |cv_{i^*}|+d_P(v_{i^*},v_1)=|c^*v_{j^*}|+|v_{j^*}v_{i^*}|+d_P(v_{i^*},v_1)=r^*$.
	For each $j\in [k^*,j^*-1]$,
\begin{equation*}
\begin{split}
d_{G(i^*,j^*-1)}(c,v_j) &\leq |cv_{j^*-1}|+d_P(v_{j^*-1},v_j) \leq 	|cv_{j^*-1}|+d_P(v_{j^*-1},v_{k^*})\\
&=|v_{i^*}v_{j^*-1}|-|cv_{i^*}|+d_P(v_{j^*-1},v_{k^*})\\
& \leq 	|v_{i^*}v_{j^*}|+|v_{j^*}v_{j^*-1}|-|cv_{i^*}|+d_P(v_{j^*-1},v_{k^*})\\
& =|v_{i^*}v_{j^*}|+|v_{j^*}v_{j^*-1}|-(|c^*v_{j^*}|+|v_{j^*}v_{i^*}|)+d_P(v_{j^*-1},v_{k^*})\\
& =|v_{j^*}v_{j^*-1}|-|c^*v_{j^*}|+d_P(v_{j^*-1},v_{k^*})\\
&=|c^*v_{j^*-1}|+d_P(v_{j^*-1},v_{k^*})=r^*.
\end{split}
\end{equation*}
For each $j\in [j^*,n]$, since $d_P(v_{k^*},v_{j^*-1})\geq d_P(v_{j^*-1},v_n)$,  $d_{G(i^*,j^*-1)}(c,v_j)\leq |cv_{j^*-1}|+d_P(v_{j^*-1},v_j)\leq |cv_{j^*-1}|+d_P(v_{j^*-1},v_n)\leq |cv_{j^*-1}|+d_P(v_{j^*-1},v_{k^*})\leq r^*$ (this last inequality was already proved above).
The claim thus follows.

The claim implies that $(i^*,j^*-1)$ is also an optimal solution with center $c$ at $e(v_{i^*},v_{j^*})$. Since $|cv_{i^*}|>0$, $c$ cannot be at $v_{i^*}$.
If $c$ is at $v_{j^*-1}\in P(v_{j^*-1},v_n)$, then this is a configuration of Case 0. Otherwise, the optimal solution $(i^*,j^*-1)$ has the same configuration as Case 1.2.1.1 (e.g., the one in Fig.~\ref{fig:config20}). In either case, this contradicts with our assumption that none of the previously discussed configurations has an optimal solution.

This proves the third part of the lemma and thus the entire lemma.
\qed
\end{proof}

Based on Lemma~\ref{lem:75}, our algorithm for this case works as follows.

For each index $i\in [1,n]$, define $k(i)$ and $j(i)$ in the same way
as in the above Case 1.2.1.1. We also compute them in $O(n)$ time.
For each $i$, if $j(i)\leq n$, then let
$r(i)=(d_P(v_1,v_i)+|v_iv_{j(i)}|+d_P(v_{j(i)},v_{k(i)}))/2$, and if
$d_P(v_1,v_i)+|v_iv_{j(i)}|<r(i)$ (implies that the center is on
$P(v_{k(i)},v_{j(i)})$), then we have a candidate solution $(i,j(i))$ with
$r(i)$ as the radius. Finally, among the at most $n$ candidate
solutions, we keep the one with the smallest radius as our solution
for this case. The total running time of the algorithm is $O(n)$.

\subsubsection{Case 1.2.2: $\pi^*$ does not contain $e(i^*,j^*)$.}

In this case, the shortest path from $c^*$ to $v_1$ in $G(i^*,j^*)$ is
$P(v_1,c^*)$ and the shortest path from $c^*$ to $v_{b^*}$ is
$P(c^*,v_{b^*})$. Since $c^*$ is in the middle of $\pi^*$, $\pi^*$ is $P(v_1,v_{b^*})$ (e.g., see Fig.~\ref{fig:config122}). Further, it is not difficult to see
that for any $j\in [b^*+1,n]$, the shortest path from $c^*$ to $v_j$ in $G(i^*,j^*)$
is $P(c^*,v_{i^*})\cup e(v_{i^*},v_{j^*})\cup P(v_{j^*},v_j)$. Also
note that $b^*<j^*$, since otherwise (i.e., $b^*=j^*$, which is
smaller than $n$ as $b^*\neq n$)
$d_{G(i^*,j^*)}(c^*,v_n)=d_{G(i^*,j^*)}(c^*,v_{b^*})+d_P(v_{b^*},v_n)>d_{G(i^*,j^*)}(c^*,v_{b^*})=r^*$,
a contradiction.
We make an assumption on $j^*$ that no index $j<j^*$ exists such that
$(i^*,j)$ is also an optimal solution with the same configuration as
$(i^*,j^*)$. We also assume that none of the previously discussed
configurations has an optimal solution.
We begin with the following observation.

\begin{figure}[t]
\begin{minipage}[t]{\textwidth}
\begin{center}
\includegraphics[height=0.7in]{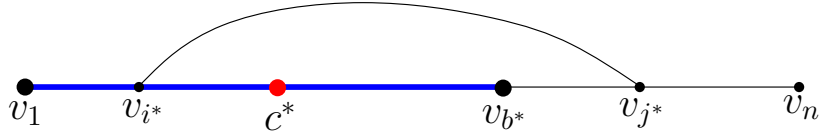}
\caption{\footnotesize  Illustrating the configuration for Case 1.2.2, where $a^*=1$, $b^*\in [i^*,j^*]$, and $c^*\in P(v_{i^*},v_{j^*})$. The thick (blue) path is $\pi^*$.}
\label{fig:config122}
\end{center}
\end{minipage}
\vspace{-0.15in}
\end{figure}

\begin{observation}\label{obser:50}
For any $i\in [1,n]$, the value $|v_iv_j|+d_P(v_j,v_n)$ is monotonically decreasing as $j$ increases from $i$ to $n$.
\end{observation}
\begin{proof}
Indeed, $|v_iv_j|+d_P(v_j,v_n)-(|v_iv_{j+1}|+d_P(v_{j+1},v_n))=|v_iv_j|+|v_iv_{j+1}|-|v_iv_{j+1}|$,
which is nonnegative by the triangle inequality. \qed
\end{proof}

Our algorithm is based on the following lemma.
\begin{lemma}\label{lem:80}
\begin{enumerate}
\item
Let $k^*$ be the largest index in $[i^*,j^*]$
such that $d_P(v_1,v_{i^*})< |v_{i^*}v_{j^*}|+d_P(v_{j^*},v_{k^*})$.
Such an index $k^*$ must exist.
\item
$b^*={k^*}$.
\item
$j^*$ must be the smallest index $j\in [i^*,n]$ such that $d_P(v_1,v_{i^*})\geq |v_{i^*}v_{j}|+d_P(v_j,v_n)$.
\item
$d_P(v_1,v_{i^*})<d_P(v_{i^*},v_n)$.
\end{enumerate}
\end{lemma}
\begin{proof}
For the first part of the lemma, it is sufficient to prove the following:
$k^*$ is the largest index in $[j',j^*]$
such that $d_P(v_1,v_{i^*})< |v_{i^*}v_{j^*}|+d_P(v_{j^*},v_{k^*})$, where $j'$ is the index of the first vertex to the right of or at $c^*$. Note that as $c^*$ is in the interior of $P(v_{i^*},v_{j^*})$, $j'\leq j^*$.
Assume to the contrary that such an index $k^*$ does not exist in $[j',j^*]$. Then, $d_P(v_1,v_{i^*})\geq  |v_{i^*}v_{j^*}|+d_P(v_{j^*},v_j)$ for all $j\in [j',j^*]$.

Let $c$ be a point on $P(v_{i^*},v_{j^*})$ arbitrarily close to $c^*$ and to the left of $c^*$. Such a point $c$ must exist as $c^*$ is in the interior of $P(v_{i^*},v_{j^*})$.
We claim that the distances from $c$ to all vertices in $G(i^*,j^*)$ are all smaller than $r^*$, which would incur contradiction. Indeed, for any $j\in [1,j'-1]$, $d_{G(i^*,j^*)}(c,v_j)\leq d_P(c,v_j)\leq d_P(c,v_1)<d_P(c^*,v_1)=r^*$. For any $j\in [j',j^*]$, $d_{G(i^*,j^*)}(c,v_j)\leq d_P(c,v_{i^*})+|v_{i^*}v_{j^*}|+d_P(v_{j^*},v_j)\leq d_P(c,v_{i^*})+d_P(v_1,v_{i^*})=d_P(c,v_1)<d_P(c^*,v_1)=r^*$.
For any $j\in [j^*+1,n]$, $d_{G(i^*,j^*)}(c,v_j)\leq d_P(c,v_{i^*})+|v_{i^*}v_{j^*}|+d_P(v_{j^*},v_n)<d_P(c^*,v_{i^*})+|v_{i^*}v_{j^*}|+d_P(v_{j^*},v_n)$.
Note that since $v_{b^*}$ is a farthest point of $c^*$, the shortest path from $c^*$ to $v_n$ must contain the new edge $e(v_{i^*},v_{j^*})$ and thus its length is $d_P(c^*,v_{i^*})+|v_{i^*}v_{j^*}|+d_P(v_{j^*},v_n)$, which is at most $r^*$. Therefore, we obtain that $d_{G(i^*,j^*)}(c,v_j)< r^*$ for each $j\in [j^*+1,n]$.
 The above claim is thus proved.


For the second part of the lemma, since the shortest path from $c^*$ to $v_{b^*}$ is $P(c^*,v_{b^*})$, it holds that $r^*=d_P(c^*,v_{b^*})\leq d_P(c^*,v_{i^*})+|v_{i^*}v_{j^*}|+d_P(v_{j^*},v_{b^*})$. Note that $r^*=d_P(c^*,v_{i^*})+|v_{i^*}v_{j^*}|+d_P(v_{j^*},v_{b^*})$ is not possible. Indeed, it this were true, then if we move $c^*$ slightly towards $v_{i^*}$, the distances from $c^*$ to all vertices in $G(i^*,j^*)$ would be all smaller than $r^*$, incurring contradiction. Hence, $r^*<d_P(c^*,v_{i^*})+|v_{i^*}v_{j^*}|+d_P(v_{j^*},v_{b^*})$. Further, since $r^*=d_P(c^*,v_1)=d_P(c^*,v_{i^*})+d_P(v_{i^*},v_1)$, we obtain that $d_P(v_{i^*},v_1)<|v_{i^*}v_{j^*}|+d_P(v_{j^*},v_{b^*})$. The definition of $k^*$ implies that $k^*\geq b^*$.

Assume to the contrary that $k^*>b^*$. Then, $d_{G(i^*,j^*)}(c^*,v_{k^*})=\min\{d_P(c^*,v_{k^*}),d_P(c^*,v_{i^*})+|v_{i^*}v_{j^*}|+d_P(v_{j^*},v_{k^*})\}$.
On the one hand, since $k^*>b^*$, $d_P(c^*,v_{k^*})>d_P(c^*,v_{b^*})=r^*$. On the other hand, since $d_P(v_{i^*},v_1)<|v_{i^*}v_{j^*}|+d_P(v_{j^*},v_{k^*})$, we obtain that $d_P(c^*,v_{i^*})+|v_{i^*}v_{j^*}|+d_P(v_{j^*},v_{k^*})>d_P(c^*,v_{i^*})+d_P(v_{i^*},v_1)=d_P(c^*,v_1)=r^*$. Hence, we derive $d_{G(i^*,j^*)}(c^*,k^*)>r^*$, a contradiction.

We proceed to prove the third statement of the lemma.
Since $d_{G(i^*,j^*)}(c^*,v_n)=d_P(c^*,v_{i^*})+|v_{i^*}v_{j^*}|+d_P(v_{j^*},v_n)\leq r^* = d_P(c^*,v_1)=d_P(c^*,v_{i^*})+d_P(v_{i^*},v_1)$, $|v_{i^*}v_{j^*}|+d_P(v_{j^*},v_n)\leq d_P(v_{i^*},v_1)$ holds.
Assume to the contrary that the third statement of the lemma is not true. Then, by Observation~\ref{obser:50}, $d_P(v_1,v_{i^*})\geq |v_{i^*}v_{j^*-1}|+d_P(v_{j^*-1},v_n)$.

Recall that $b^*<j^*$, and thus $k^*=b^*\leq j^*-1$.
We claim that the distances from $c^*$ to all vertices in $G(i^*,j^*-1)$ are no more than $r^*$.
Indeed, for each $j\in [1,k^*]$, $d_{G(i^*,j^*-1)}(c^*,v_j)= d_P(c^*,v_{j})\leq d_P(c^*,v_{1})= r^*$.
If $k^*<j^*-1$, then for each $j\in [k^*+1,j^*-1]$, $d_{G(i^*,j^*-1)}(c^*,v_j)\leq d_P(c^*,v_{i^*})+|v_{i^*}v_{j^*-1}|+d_P(v_{j^*-1},v_{j})\leq d_P(c^*,v_{i^*})+|v_{i^*}v_{j^*}|+|v_{j^*}v_{j^*-1}|+d_P(v_{j^*-1},v_{j})
=d_P(c^*,v_{i^*})+|v_{i^*}v_{j^*}|+d_P(v_{j^*},v_{j})$, which is the shortest path length from $c^*$ to $v_j$ in $G(i^*,j^*)$ and thus is at most $r^*$.
For each $j\in [j^*,n]$, because $d_P(v_1,v_{i^*})\geq |v_{i^*}v_{j^*-1}|+d_P(v_{j^*-1},v_n)$, $d_{G(i^*,j^*-1)}(c^*,v_j)\leq d_P(c^*,v_{i^*})+|v_{i^*}v_{j^*-1}|+d_P(v_{j^*-1},v_{j})\leq d_P(c^*,v_{i^*})+|v_{i^*}v_{j^*-1}|+d_P(v_{j^*-1},v_{n})\leq d_P(c^*,v_{i^*})+d_P(v_{i^*},v_{1})=r^*$.
The claim is thus proved.
The claim implies that $(i^*,j^*-1)$ is also an optimal solution with the same configuration as $(i^*,j^*)$. But this contradicts with our assumption on $j^*$: no index $j<j^*$ exists such that
$(i^*,j)$ is also an optimal solution with the same configuration as
$(i^*,j^*)$.

For the fourth part of the lemma, since $c^*$ is strictly to the right of $v_{i^*}$, $d_P(v_{i^*},v_n)\geq d_P(v_{i^*},v_{b^*})>d_P(c^*,v_{b^*})=r^*=d_P(c^*,v_1)>d_P(v_{i^*},v_1)$.
\qed
\end{proof}

Based on Lemma~\ref{lem:80}, our algorithm works as follows.

Let $i_1$ be the largest index $i$ in $[1,n]$ such that $d_P(v_1,v_{i})<d_P(v_{i},v_n)$.
Let $i_2$ be the smallest index $i$ in $[1,n]$ such that $d_P(v_1,v_i)\geq |v_iv_n|$.
By Lemma~\ref{lem:80} and Observation~\ref{obser:50}, if $i_2\leq i_1$\footnote{Note that $i_2\leq i_1+1$ always holds because $d_P(v_1,v_{i_1+1})\geq d_P(v_{i_1+1},v_n)\geq |v_{i_1+1}v_n|$.}, we only need to consider the indices in $[i_2,i_1]$ as the candidates for $i^*$. For each $i\in [i_2,i_1]$, define $j(i)$ as the smallest index $j\in [i,n]$ such that $d_P(v_1,v_i)\geq |v_iv_j|+d_P(v_j,v_n)$\footnote{The index $j$ must exist because $d_P(v_1,v_i)\geq |v_iv_n|$ due to the definition of $i_2$.}, and
define $k(i)$ as the largest $k\in [i,j(i)]$ such that $d_P(v_1,v_i)<|v_iv_{j(i)}|+ d_P(v_{j(i)},v_k)$ (for convenience let $k(i)=0$ if no such index $k$ exists).

The monotonicity properties of $j(i)$ and $k(i)$ in the following lemma will lead to an efficient algorithm to compute them.
\begin{lemma}\label{obser:70}
For any $i\in [i_2,i_1-1]$, $j(i+1)\leq j(i)$ and $k(i+1)\leq k(i)$.
\end{lemma}
\begin{proof}
To prove that $j(i+1)\leq j(i)$,
it is sufficient to show that for any $j\in [i+1,n]$, if $d_P(v_1,v_i)\geq |v_iv_j|+d_P(v_j,v_n)$, then $d_P(v_1,v_{i+1})\geq |v_{i+1}v_j|+d_P(v_j,v_n)$. Indeed, by the triangle inequality,
\begin{equation*}
\begin{split}
d_P(v_1,v_{i+1}) = d_P(v_1,v_{i}) + |v_iv_{i+1}| \geq |v_iv_j|+d_P(v_j,v_n) + |v_iv_{i+1}|\geq |v_{i+1}v_j|+d_P(v_j,v_n).
\end{split}
\end{equation*}

To prove $k(i+1)\leq k(i)$, since $[i+1,j(i+1)]\subseteq [i,j(i)]$, it is sufficient to show that for any $k\in [i+1,j(i+1)]$, if $d_P(v_1,v_{i+1})<|v_{i+1}v_{j(i+1)}|+ d_P(v_{j(i+1)},v_k)$, then $d_P(v_1,v_i)<|v_iv_{j(i)}|+ d_P(v_{j(i)},v_k)$. Indeed, due to the triangle inequality, we have
\begin{equation*}
\begin{split}
d_P(v_1,v_i) & = d_P(v_1,v_{i+1}) - |v_iv_{i+1}| < |v_{i+1}v_{j(i+1)}|+ d_P(v_{j(i+1)},v_k) - |v_iv_{i+1}|\\
& \leq |v_iv_{j(i+1)}| + d_P(v_{j(i+1)},v_k) = |v_iv_{j(i+1)}| + d_P(v_{j(i)},v_k) - d_P(v_{j(i+1)},v_{j(i)})\\
& \leq |v_iv_{j(i)}| +  d_P(v_{j(i)},v_k).
\end{split}
\end{equation*}
The lemma thus follows.
\qed
\end{proof}

Our algorithm for this configuration works as follows.
We first compute the two indices $i_1$ and $i_2$. If $i_2>i_1$, then we do not keep any solution for this case. Otherwise, by the monotonicity properties of $j(i)$ and $k(i)$ in Lemma~\ref{obser:70}, we can compute $j(i)$ and $k(i)$ for all $i\in [i_2,i_1]$ in $O(n)$ time by a linear scan on $P$. The details are omitted.
Then, for each $i\in [i_2,i_1]$, if $k(i)\neq 0$ and $d_P(v_1,v_i)<d_P(v_i,v_{k(i)})$ (this makes sure that the center is in $P(v_i,v_{k(i)})$), then we have a candidate solution $(i,j(i))$ with radius $r(i)=d_P(v_1,v_{k(i)})/2$. By our definition of $j(i)$ and $k(i)$, the solution is feasible. Finally, among all the at most $n$ candidate solutions, we keep the one with the smallest radius as our solution for this case. The algorithm runs in $O(n)$ time.

\subsubsection{Case 2: $b^*=n$.}

This case is symmetric to Case 1 ($a^*=1$), so we omit the details.

\subsubsection{Case 3: Both $a^*$ and $b^*$ are in $[i^*,j^*]$.}

Observe that since $a^*\neq 1$ and $b^*\neq n$,
$a^*$ cannot be $i^*$ and $b^*$ cannot be $j^*$. Hence, both $a^*$ and $b^*$ are in $[i^*+1,j^*-1]$.
As in Case 1, depending on whether $c^*$ is in $e(i^*,j^*)$ or $P(v_{i^*},v_{j^*})$, there are two subcases.

\subsubsection{Case 3.1. $c^*\in e(i^*,j^*)$.}

More precisely, $c^*$ is in the interior of $e(i^*,j^*)$, which implies that $e(i^*,j^*)$ is in $\pi^*$. It is not difficult to see that $b^*=a^*+1$ (e.g., see Fig.~\ref{fig:config31}). We make an assumption on $[i^*,j^*]$ that there is no smaller interval $[i,j]\subset [i^*,j^*]$ such that $(i,j)$ is also an optimal solution with the same configuration as $(i^*,j^*)$ (since otherwise we could instead consider $(i,j)$ as  $(i^*,j^*)$).

\begin{figure}[t]
\begin{minipage}[t]{\textwidth}
\begin{center}
\includegraphics[height=0.7in]{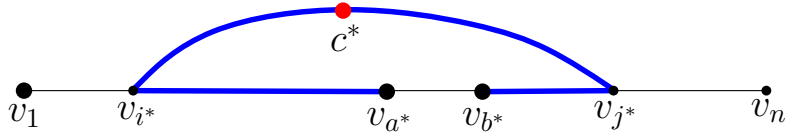}
\caption{\footnotesize  Illustrating the configuration for Case 3.1, where $a^*,b^*\in [i^*,j^*]$ and $c^*\in e(v_{i^*},v_{j^*})$. The thick (blue) path is $\pi^*$.}
\label{fig:config31}
\end{center}
\end{minipage}
\vspace{-0.15in}
\end{figure}

We also assume that none of the previously discussed cases happens. This implies that neither $v_1$ nor $v_n$ is a farthest vertex of $c^*$ in $G(i^*,j^*)$. To see this, suppose to the contrary that $v_1$ is also a farthest vertex. Then, if we consider $v_1$ and $v_{b^*}$ as two farthest vertices stated in Observation~\ref{obser:10}, then the configuration becomes Case 1.2.1.1 (shown in the bottom example for Fig.~\ref{fig:config}), which incurs contradiction. Similarly, $v_n$ is not a farthest vertex as well.
Since neither $v_1$ nor $v_n$ is a farthest vertex of $c^*$, it can be verified that $d_P(v_1,v_{i^*})<  d_P(v_{i^*},v_{a^*})$ and $d_P(v_{j^*},v_n)< d_P(v_{b^*},v_{j^*})$.

\begin{lemma}\label{lem:90}
$i^*$ is the largest index $i\in [1,a^*]$ such that $d_P(v_1,v_i)< d_P(v_i,v_{a^*})$.
$j^*$ is the smallest index $j\in [b^*,n]$ such that $d_P(v_j,v_n)< d_P(v_{b^*},v_j)$.
\end{lemma}
\begin{proof}
We only prove the first part of the lemma, as the proof for the second part is analogous.
Assume to the contrary that this is not true. Then, due to $d_P(v_1,v_{i^*})<  d_P(v_{i^*},v_{a^*})$, it must hold that
	$d_P(v_1,v_{i^*+1})< d_P(v_{i^*+1},v_{a^*})$.

Let $c$ be a point on $e(v_{i^*},v_{i^*+1})\cup e(v_{i^*+1},v_{j^*})$
	with distance $|c^*v_{j^*}|$ from $v_{j^*}$ (e.g., see Fig.~\ref{fig:newgraph50}). Such a point $c$ must exist
	since $|v_{i^*}v_{i^*+1}|+|v_{i^*+1}v_{j^*}|\geq
	|v_{i^*}v_{j^*}|$. Depending on whether $c\in e(v_{i^*+1},v_{j^*})$,
	there are two cases.

\begin{figure}[t]
\begin{minipage}[t]{0.49\textwidth}
\begin{center}
\includegraphics[height=0.65in]{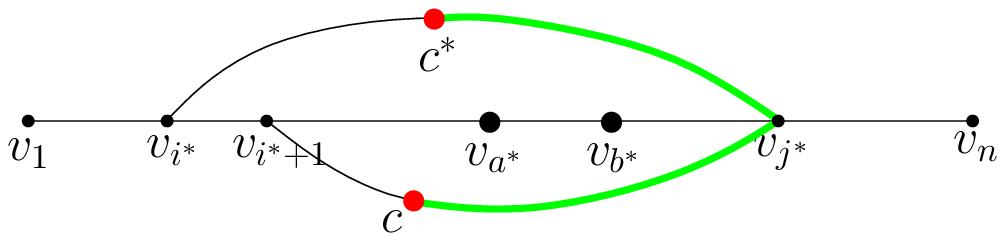}
\caption{\footnotesize  Illustrating the proof of Lemma~\ref{lem:90}: The two (green) thick curves from $v_{j^*}$ have the same length, where $c\in e(v_{i^*+1},v_{j^*})$.}
\label{fig:newgraph50}
\end{center}
\end{minipage}
\hspace{0.05in}
\begin{minipage}[t]{0.49\textwidth}
\begin{center}
\includegraphics[height=0.65in]{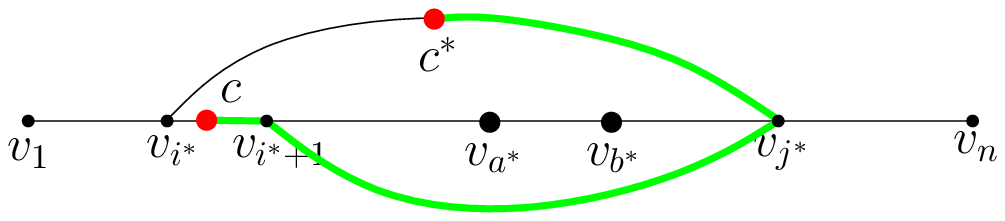}
\caption{\footnotesize  Illustrating the proof of Lemma~\ref{lem:90}: The two (green) thick curves from $v_{j^*}$ have the same length, where $c\in e(v_{i^*},v_{i^*+1})$.}
\label{fig:newgraph60}
\end{center}
\end{minipage}
\vspace{-0.15in}
\end{figure}

If $c\in e(v_{i^*+1},v_{j^*})$ (e.g., see Fig.~\ref{fig:newgraph50}), then we claim that the distances from
	$c$ to all vertices in $G(i^*+1,j^*)$ are at most
	$r^*$. Indeed, for each $j\in [b^*,j^*]$, $d_{G(i^*+1,j^*)}(c,v_j)\leq
	|cv_{j^*}|+d_P(v_{j^*},v_j)\leq
	|cv_{j^*}|+d_P(v_{j^*},v_{b^*})=|c^*v_{j^*}|+d_P(v_{j^*},v_{b^*})=r^*$.
	For each $j\in [j^*+1,n]$, $d_{G(i^*+1,j^*)}(c,v_j)\leq |cv_{j^*}|+d_P(v_{j^*},v_j)\leq |cv_{j^*}|+d_P(v_{j^*},v_n)<|cv_{j^*}|+d_P(v_{j^*},v_{b^*})=|c^*v_{j^*}|+d_P(v_{j^*},v_{b^*})=r^*$.
For each $j\in [i^*+1,a^*]$, $d_{G(i^*+1,j^*)}(c,v_j)\leq
	|cv_{i^*+1}|+d_P(v_{i^*+1},v_j)\leq |cv_{i^*+1}|+d_P(v_{i^*+1},v_{a^*})$. In a similar way as
	Equation~\eqref{equ:30}, we can show that
	$|cv_{i^*+1}|+d_P(v_{i^*+1},v_{a^*})\leq
	|c^*v_{i^*}|+d_P(v_{i^*},v_{a^*})=r^*$. For each $j\in [1,i^*]$,
	due to $d_P(v_1,v_{i^*+1})< d_P(v_{i^*+1},v_{a^*})$,
	$d_{G(i^*+1,j^*)}(c,v_j)\leq |cv_{i^*+1}|+d_P(v_{i^*+1},v_j)\leq
	|cv_{i^*+1}|+d_P(v_{i^*+1},v_1)<
	|cv_{i^*+1}|+d_P(v_{i^*+1},v_{a^*})\leq r^*$.
The claim is thus proved. However, the claim implies that
	$(i^*+1,j^*)$ is also an optimal solution with the same
	configuration as $(i^*,j^*)$. As $[i^*+1,j^*]\subset [i^*,j^*]$, we obtain
	contradiction with our assumption on $[i^*,j^*]$: there is no smaller interval $[i,j]\subset [i^*,j^*]$ such that $(i,j)$ is also an optimal solution with the same configuration as $(i^*,j^*)$.

If $c\not\in e(v_{i^*+1},v_{j^*})$ (e.g., see Fig.~\ref{fig:newgraph60}), then $c$ is on
	$e(v_{i^*},v_{i^*+1})\setminus\{v_{i^*+1}\}$. Let $c'=v_{i^*+1}$.
	We claim that the distances from $c'$ to all
	vertices in $G(i^*+1,j^*)$ are strictly smaller than $r^*$,
	which incurs contradiction. Indeed, for each $j\in [b^*,n]$,
	$d_{G(i^*+1,j^*)}(c',v_j)\leq |c'v_{j^*}|+d_P(v_j,v_{j^*})\leq
	|c'v_{j^*}|+d_P(v_{b^*},v_{j^*})<|cc'|+|c'v_{i^*+1}|+|v_{i^*+1}v_{j^*}|+d_P(v_{j^*},v_{b^*})=r^*$.
	Since $|c^*v_{i^*}|+d_P(v_{i^*},v_{i^*+1})+d_P(v_{i^*+1},v_{a^*})=r^*$,
	we obtain that $d_P(c',v_{a^*})<r^*$.
	Hence, for each $j\in [1,a^*]$, since $d_P(v_1,v_{i^*+1})<
	d_P(v_{i^*+1},v_{a^*})$, we obtain $d_{G(i^*+1,j^*)}(c',v_j)\leq
	d_P(c',v_{a^*})<r^*$.
	\qed
\end{proof}

Based on Lemma~\ref{lem:90}, our algorithm works as follows.
For each interval $[k,k+1]$ with $k\in [2,n-2]$ (since $a^*>1$ and $b^*<n$, we do not need to consider the case where $k=1$ or $k+1=n$), define $i(k)$ as the largest index $i\in [1,k]$ such that $d_P(v_1,v_i)<d_P(v_i,v_k)$, and define $j(k)$ as the smallest index $j\in [k+1,n]$ such that $d_P(v_j,v_n)<d_P(v_{k+1},v_j)$. It can be verified that
for any $k\in [2,n-3]$, $i(k)\leq i(k+1)$ and $j(k)\leq j(k+1)$.
Thus, we can easily compute $i(k)$ and $j(k)$ for all $k\in [2,n-2]$ in $O(n)$ time. Then, for each $k\in [2,n-2]$, let $r(i)=(d_P(v_{i(k)},v_k)+|v_{i(k)}v_{j(k)}|+d_P(v_{k+1},v_{j(k)}))/2$, and if $r(i)>d_P(v_{i(k)},v_k)$ and $r(i)>d_P(v_{k+1},v_{j(k)})$ (this makes sure that the center is on the edge $e(v_{i(k)},v_{j(k)})$), then we have a candidate solution $(i(k),j(k+1))$ with $r(i)$ as the radius. By the definitions of $i(k)$ and $j(k+1)$, the solution is feasible. Finally, among the at most $n$ candidate solutions, we keep the one with the smallest radius as our solution for this case. The total time of the algorithm is $O(n)$.

\subsubsection{Case 3.2. $c^*\in P(v_{i^*},v_{j^*})$.}
More precisely, $c^*$ is in the interior of $P(v_{i^*},v_{j^*})$.
We first have the following observation.

\begin{figure}[t]
\begin{minipage}[t]{\textwidth}
\begin{center}
\includegraphics[height=0.7in]{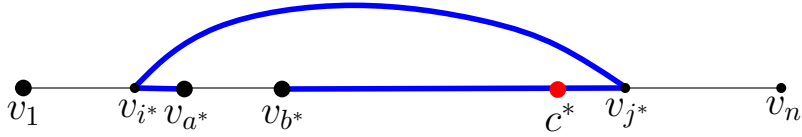}
\caption{\footnotesize  Illustrating the configuration for Case 3.2, where $a^*,b^*\in [i^*,j^*]$ and $c^*\in P(v_{i^*},v_{j^*})$. The thick (blue) path is $\pi^*$.}
\label{fig:config32}
\end{center}
\end{minipage}
\vspace{-0.15in}
\end{figure}

\begin{lemma}\label{obser:80}
$\pi^*$ must contain $e(v_{i^*},v_{j^*})$; $b^*=a^*+1$;
 $v_{a^*}$ and $v_{b^*}$ are on the same side of $c^*$ (e.g., see Fig.~\ref{fig:config32}).
\end{lemma}
\begin{proof}
Assume to the contrary $\pi^*$ does not contain $e(v_{i^*},v_{j^*})$. Then, $c^*$ is between $v_{a^*}$ and $v_{b^*}$, and $\pi^*=P(v_{a^*},v_{b^*})$, e.g., see Fig.~\ref{fig:obser10}. Recall that $a^*>1$, $b^*<n$, and both $a^*$ and  $b^*$ are in $[i^*+1,j^*-1]$. Since $d_P(c^*,v_{a^*})=d_P(c^*,v_{b^*})=r^*$, one can verify that the distance from $c^*$ to $v_1$ (or $v_n$) in $G(i^*,j^*)$ must be larger than $r^*$, which incurs contradiction.
\begin{figure}[t]
\begin{minipage}[t]{\textwidth}
\begin{center}
\includegraphics[height=0.7in]{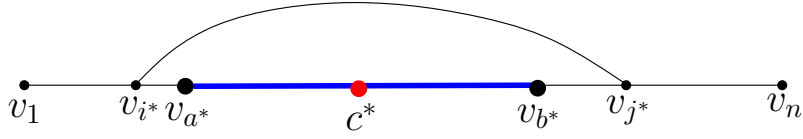}
\caption{\footnotesize  Illustrating the proof of Lemma~\ref{obser:80}. The thick (blue) path is $\pi^*$.}
\label{fig:obser10}
\end{center}
\end{minipage}
\vspace{-0.15in}
\end{figure}

Since $\pi^*$ contains $e(v_{i^*},v_{j^*})$, $v_{a^*}$ and $v_{b^*}$ must be on the same side of $c^*$, in which case $v_{a^*}$ and $v_{b^*}$ must be two adjacent vertices, i.e., $b^*=a^*+1$.
\qed
\end{proof}

In the following, we only discuss the case where $c^*$ is to the right
of $v_{a^*}$ and $v_{b^*}$ (e.g., see Fig.~\ref{fig:config32}),  and the algorithm for the other case
is symmetric. We make an assumption on $[i^*,j^*]$ that there is no smaller interval $[i,j]\subset [i^*,j^*]$ such that $(i,j)$ is also an optimal solution with the same configuration as $(i^*,j^*)$.
We again assume that none of the previously discussed cases happens.
The following lemma is literally the same as Lemma~\ref{lem:90} although the proof
is different.

\begin{lemma}\label{lem:100}
$i^*$ is the largest index $i\in [1,a^*]$ such that $d_P(v_1,v_i)< d_P(v_i,v_{a^*})$.
$j^*$ is the smallest index $j\in [b^*,n]$ such that $d_P(v_j,v_n)< d_P(v_j,v_{b^*})$.
\end{lemma}
\begin{proof}
Unlike Lemma~\ref{lem:90}, the proofs here for the two parts of the
	lemma are different.
We begin with the first statement for $i^*$.

Since $v_1$ is not a farthest vertex of $c^*$ in $G(i^*,j^*)$, we claim that $d_P(v_1,v_{i^*})< d_P(v_{i^*},v_{a^*})$.
Indeed, since $c^*$ is to the right of $b^*$ and $\pi^*$ contains $e(v_{i^*},v_{j^*})$, $d_{G(i^*,j^*)}(c^*,v_{a^*})=d_P(c^*,v_{j^*})+|v_{j^*}v_{i^*}|+d_P(v_{i^*},v_{a^*})=r^*$ and $d_{G(i^*,j^*)}(c^*,v_1)=d_P(c^*,v_{j^*})+|v_{j^*}v_{i^*}|+d_P(v_{i^*},v_{1})$. Since $v_1$ is not a farthest point of $c^*$, $d_{G(i^*,j^*)}(c^*,v_1)<r^*=d_{G(i^*,j^*)}(c^*,v_{a^*})$, which leads to $d_P(v_1,v_{i^*})< d_P(v_{i^*},v_{a^*})$.

Assume to the contrary that $i^*$ is not the largest such index as stated in the lemma. Then, we must have $d_P(v_1,v_{i^*+1})< d_P(v_{i^*+1},v_{a^*})$. We claim that the distances from $c^*$ to all vertices in $G(i^*+1,j^*)$ are at most $r^*$. This can be proved by the similar argument as that for Lemma~\ref{lem:90}, so we omit the details. The claim implies that $(i^*+1,j^*)$ is also an optimal solution with the same configuration as $(i^*,j^*)$. Since $[i^*+1,j^*]\subset [i^*,j^*]$, this contradicts with our assumption on $[i^*,j^*]$.

We proceed to prove the second statement of the lemma for $j^*$.

Since $v_n$ is not a farthest vertex of $c^*$ in $G(i^*,j^*)$, we claim that $d_P(v_{j^*},v_n)< d_P(v_{j^*},v_{b^*})$. Indeed, since $c^*$ is to the right of $v_{b^*}$ and $\pi^*$ contains $e(v_{i^*},v_{j^*})$, $d_{G(i^*,j^*)}(c^*,v_n)=d_P(c^*,v_n)$. Because $v_n$ is not a farthest point of $c^*$ in $G(i^*,j^*)$, $d_{G(i^*,j^*)}(c^*,v_n)<r^*$. Thus, $d_P(c^*,v_n)<r^*$.
Since $v_{j^*}\in P(c^*,v_{n})$, $d_P(v_{j^*},v_n)\leq  d_P(c^*,v_n)<r^* = d_P(c^*,v_{b^*})\leq d_P(v_{j^*},v_{b^*})$.


Next, we show that $c^*$ must be in the interior of $e(v_{j^*-1},v_{j^*})$. Assume to the contrary this is not true. Then,  since $c^*\neq v_{j^*}$, $c^*$ is in $P(v_{b^*,},v_{j^*-1})$ (e.g., see Fig.~\ref{fig:newgraph70}).
We claim that the distances from $c^*$ to all vertices in $G(i^*,j^*-1)$ are at most $r^*$. Indeed, for each vertex $v_j$ to the right of $c^*$, $d_{G(i^*,j^*-1)}(c^*,v_j)\leq d_P(c^*,v_j)\leq d_P(c^*,v_n)<r^*$. For each vertex $v_j$ to the left of $c^*$ but to the right of $v_{b^*}$, $d_{G(i^*,j^*-1)}(c^*,v_j)\leq d_P(c^*,v_j)\leq d_P(c^*,v_{b^*})=r^*$. For each $j\in [1,a^*]$, since $d_P(v_1,v_{i^*})< d_P(v_{i^*},v_{a^*})$ (which was proved above), we have
\begin{equation*}
\begin{split}
d_{G(i^*,j^*-1)}(c^*,v_j) &\leq d_P(c^*,v_{j^*-1})+|v_{j^*-1}v_{i^*}|+d_P(v_{i^*},v_{j})\\
& \leq d_P(c^*,v_{j^*-1})+|v_{j^*-1}v_{i^*}|+d_P(v_{i^*},v_{a^*})\\
& \leq d_P(c^*,v_{j^*})+|v_{i^*}v_{j^*}|+d_P(v_{i^*},v_{a^*}) = r^*.
\end{split}
\end{equation*}
This proves the claim. The claim implies that $(i^*,j^*-1)$ is also an optimal solution with the same configuration as $(i^*,j^*)$. Since $[i^*,j^*-1]\subset [i^*,j^*]$, this contradicts with our assumption on $[i^*,j^*]$.

\begin{figure}[t]
\begin{minipage}[t]{\textwidth}
\begin{center}
\includegraphics[height=0.7in]{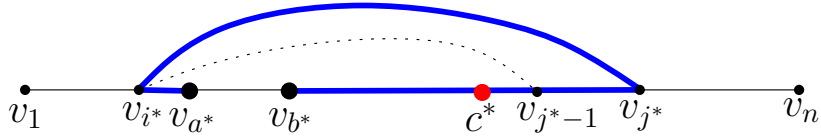}
\caption{\footnotesize  Illustrating the proof for Lemma~\ref{lem:100}: $c^*$ is to the left of $v_{j^*-1}$.}
\label{fig:newgraph70}
\end{center}
\end{minipage}
\end{figure}

We proceed to show that $j^*$ is the smallest such index as stated in the lemma. Assume to the contrary that this is not true. Then, since $d_P(v_{j^*},v_n)< d_P(v_{j^*},v_{b^*})$, which is proved above,  $d_P(v_{j^*-1},v_n)< d_P(v_{j^*-1},v_{b^*})$ must hold.
For a cycle $C(i,j)$, let $|C(i,j)|$ denote its total length.
Let $r=(|C(i^*,j^*-1)|-|v_{a^*}v_{b^*}|)/2$. Note that $r^*=(|C(i^*,j^*)|-|v_{a^*}v_{b^*}|)/2$. Due to the triangle inequality, we have $|C(i^*,j^*-1)|\leq |C(i^*,j^*)|$. Hence, $r\leq r^*$. Let $c$ be the point on $C(i^*,j^*-1)\setminus e(v_{a^*},v_{b^*})$ whose distances from $v_{a^*}$ and $v_{b^*}$ are both equal to $r$ (e.g., see Fig.~\ref{fig:new80}). Note that $c$ may be in $P(v_{i^*},v_{a^*})$, $e(v_{i^*},v_{j^*-1})\setminus\{v_{i^*},v_{j^*-1}\}$, or $P(v_{b^*},v_{j^*-1})$. In the following, we show that none of the three cases can happen.

If $c\in P(v_{i^*},v_{a^*})$, then let $c'=v_{i^*}$ and we claim that the distances from $c'$ to all vertices in $G(i^*,j^*-1)$ are smaller than $r^*$, which would incur contradiction. Indeed, since $d_P(c^*,v_{j^*})+|v_{j^*}v_{i^*}|+d_P(v_{i^*},v_{a^*})=r^*$ and $d_P(c^*,v_{j^*})>0$ (because $c^*\neq v_{j^*}$), we have $r=d_P(c,v_{a^*})\leq d_P(v_{i^*},v_{a^*})<r^*$. Hence, for each $j\in [i^*,a^*]$, $d_{G(i^*,j^*-1)}(c',v_j)\leq d_P(c',v_j)\leq d_P(c',v_{a^*})<r^*$. Since $d_P(v_1,v_{i^*})<d_P(v_{i^*},v_{a^*})$, for each $j\in [1,i^*]$, $d_{G(i^*,j^*-1)}(c',v_j)\leq d_P(c',v_j)\leq d_P(c',v_{1})< d_P(c',v_{a^*})<r^*$.
For each $j\in [b^*,j^*-1]$, $d_{G(i^*,j^*-1)}(c',v_j)\leq |v_{i^*}v_{j^*-1}|+d_P(v_{j^*-1},v_j)\leq d_P(c,v_{i^*})+|v_{i^*}v_{j^*-1}|+d_P(v_{j^*-1},v_{j})\leq d_P(c,v_{i^*})+|v_{i^*}v_{j^*-1}|+d_P(v_{j^*-1},v_{b^*})=r<r^*$. Since $d_P(v_{j^*-1},v_{b^*})>d_P(v_{j^*-1},v_n)$, for each $j\in [j^*,n]$, $d_{G(i^*,j^*-1)}(c',v_j)\leq |v_{i^*}v_{j^*-1}|+d_P(v_{j^*-1},v_j)\leq d_P(c,v_{i^*})+|v_{i^*}v_{j^*-1}|+d_P(v_{j^*-1},v_j)\leq
d_P(c,v_{i^*})+|v_{i^*}v_{j^*-1}|+d_P(v_{j^*-1},v_n)<  d_P(c,v_{i^*})+|v_{i^*}v_{j^*-1}|+d_P(v_{j^*-1},v_{b^*})=r<r^*$. The above claim is thus proved.

\begin{figure}[t]
\begin{minipage}[t]{\textwidth}
\begin{center}
\includegraphics[height=0.9in]{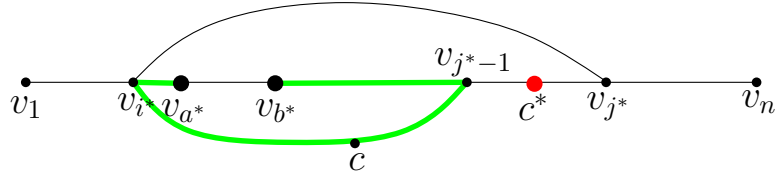}
\caption{\footnotesize  Illustrating the proof for Lemma~\ref{lem:100}: The green curve from $c$ to $v_{a^*}$ and the one from $c$ to $v_{b^*}$ have the same length.}
\label{fig:new80}
\end{center}
\end{minipage}
\vspace{-0.15in}
\end{figure}

If $c\in e(v_{i^*},v_{j^*-1})\setminus\{v_{i^*},v_{j^*-1}\}$ (e.g., see Fig.~\ref{fig:new80}), then we claim that the distances from $c$ to all vertices in $G(i^*,j^*-1)$ are at most $r^*$. Indeed, for each $j\in [1,a^*]$, $d_{G(i^*,j^*-1)}(c,v_j)\leq |cv_{i^*}|+d_P(v_{i^*},v_j)\leq |cv_{i^*}|+d_P(v_{i^*},v_{a^*})=r\leq r^*$. Similarly, for each $j\in [b^*,n]$, $d_{G(i^*,j^*-1)}(c,v_j)\leq |cv_{j^*-1}|+d_P(v_{j^*-1},v_j)\leq |cv_{j^*-1}|+d_P(v_{j^*-1},v_{b^*})=r\leq r^*$. The claim thus follows. The claim implies that $(i^*,j^*-1)$ is  an optimal solution conforming with the configuration of Case 3.1 (i.e., the one shown in Fig.~\ref{fig:config31}), which contradicts with our assumption that none of the previously discussed cases happens.

If $c\in P(v_{b^*},v_{j^*-1})$, then let $c'=v_{j^*-1}$. We claim that the distances from $c'$ to all vertices in $G(i^*,j^*-1)$ are smaller than $r^*$, which would incur contradiction. Indeed, as $c^*$ is in the interior of $e(j^*-1,j^*)$, $r= d_P(v_{b^*},c)\leq d_P(v_{b^*},c')<d_P(v_{b^*},c^*)=r^*$. For each $j\in [1,a^*]$, $d_{G(i^*,j^*-1)}(c',v_j)\leq |v_{j^*-1}v_{i^*}|+d_P(v_{i^*},v_j)\leq d_P(c,v_{j^*-1})+|v_{j^*-1}v_{i^*}|+d_P(v_{i^*},v_j) \leq
d_P(c,v_{j^*-1})+|v_{j^*-1}v_{i^*}|+d_P(v_{i^*},v_{a^*}) = r < r^*$.
For each $j\in [b^*,n]$, since $d_P(v_{j^*-1},v_n)< d_P(v_{j^*-1},v_{b^*})$, $d_{G(i^*,j^*-1)}(c',v_j)\leq d_P(c',v_j)\leq d_P(c',v_{b^*})< d_P(c^*,v_{b^*}) = r^*$. The claim thus follows.
\qed
\end{proof}

Based on Lemma~\ref{lem:100}, our algorithm for this configuration works as follows.
We define $i(k)$ and $j(k)$ for each $k\in [2,n-2]$ in the same way as those for Case 3.1, and their values have already been computed in Case 3.1. Then, for each $k\in [2,n-2]$, let $r(i)=(d_P(v_k,v_{i(k)})+|v_{i(k)}v_{j(k)}|+d_P(j(k),v_{k+1}))/2$, and if $r(i)< d_P(v_{k+1},v_{j(k)})$ (this makes sure that the center is on $P(v_{k+1},v_{j(k)}$), then we keep $(i(k),j(k))$ as a candidate solution with $r(i)$ as the radius. The definitions of $i(k)$ and $j(k)$ guarantee that it is a feasible solution. Finally, among the at most $n$ candidate solutions, we keep the one with the smallest radius as the solution for this configuration. The total time of the algorithm is $O(n)$.

\paragraph{Remark.} The above gives the algorithm for Case 3.2 when $c^*$ is to the right of $v_{b^*}$. If $c^*$ is to the left of $v_{a^*}$, then we also use the above same values $i(k)$, $j(k)$, and $r(i)$. We keep the candidate solution only if $r(i)< d_P(v_{i(k)},v_{k})$ (this makes sure that the center is on $P(v_{i(k)},v_{k}$). In fact, the can unify our algorithms for Case 3.1 and Case 3.2 to obtain an algorithm for Case 3, as follows. We compute the same values $i(k)$, $j(k)$, and $r(i)$ as before. Then, for each  $k\in [2,n-2]$, we keep $(i(k),j(k))$ as a candidate solution with $r(i)$ as the radius. Finally, among all at most $n$ candidate solutions, we keep the one with the smallest radius for Case 3.

\paragraph{Summary.}
The above provides a linear time algorithm for computing a candidate solution $(i,j)$ (along with a radius $r$ and a corresponding center $c$) for each configuration so that if there is an optimal solution that has the same configuration then $(i,j)$ is also an optimal solution with $c$ as the center and $r=r^*$. On the other hand, each such solution is feasible in the sense that the distances from $c$ to all vertices in $G(i,j)$ are at most $r$. Given an input instance, since we do not know which configuration has an optimal solution, we use the above algorithm to compute a constant number of candidate solutions, and among them, we return the one with the smallest radius. The correctness follows our discussions above. The running time of the algorithm is $O(n)$.

\begin{theorem}
The ROAP problem is solvable in linear time.
\end{theorem}

\section{The Query Algorithm}
\label{sec:query}

As a by-product of our techniques, we present an $O(\log n)$ time algorithm to compute the radius and a center of the graph $G(i,j)=P\cup e(v_i,v_j)$ for any query pair of indices $(i,j)$, after $O(n)$ time preprocessing. The result may be interesting in its own right.

We perform the linear time preprocessing as in Section~\ref{sec:pre} so that $d_P(v_i,v_j)$ can be computed in $O(1)$ time for any $(i,j)$.

In addition, we need to perform preprocessing to answer the following {\em range-maxima} queries. Given any pair $(i,j)$, find the interval $[k,k+1]$ such that the length $|v_kv_{k+1}|$ is the largest among all $k\in [i,j-1]$. The query can be answered in $O(\log n)$ time by using a binary search tree $T$ as follows. $T$ has $n-1$ leaves that correspond to the intervals $[k,k+1]$ for $k=1,2,\ldots, n-1$ respectively. The root of $T$ stores the interval $[k,k+1]$ with the largest $|v_kv_{k+1}|$ for all $k\in [1,n-1]$. The left subtree of $T$ is built with respect to $k=1,2,\ldots,\lfloor n/2\rfloor$ recursively, and the right subtree is built with respect to $k=\lfloor n/2\rfloor+1,\ldots,n-1$ recursively. $T$ can be built in $O(n)$ time in a bottom-up fashion. With $T$, each range-maxima query can be answered in $O(\log n)$ time in a standard way (e.g., like queries in segment or interval trees).

\paragraph{Remark.} Another more efficient but rather complicated method is to use range maxima data structure~\cite{ref:BenderTh00,ref:HarelFa84}, and each query can be answered in $O(1)$ time and the preprocessing time is still $O(n)$. However, since $O(\log n)$ time query is sufficient for our purpose, the above binary tree method, which is quite simple, is preferable.
\medskip

The above is our preprocessing algorithm, which runs in $O(n)$ time.
In the sequel, we present our query algorithm. Let $(i,j)$ be a query with $i\leq j$. Let $G=G(i,j)$. Denote by $c$ a center of $G$ and $r$ the radius. Note that Observation~\ref{obser:10} is still applicable (replacing $i^*$, $j^*$, $c^*$, $r^*$ by $i$, $j$, $c$, $r$, respectively). Let $a$ and $b$ respectively be $a^*$ and $b^*$ stated in Observation~\ref{obser:10}, and let $\pi$ be the union of the two paths $\pi_{a}$ and $\pi_{b}$ corresponding to $\pi_{a^*}$ and $\pi_{b^*}$ in Observation~\ref{obser:10}.
Without loss of generality, we assume that $a\leq b$.

Our query algorithm works as follows. Depending on $a$, $b$, and $\pi$, there are several possible configurations as discussed in Section~\ref{sec:algo}. For each configuration, we will compute in $O(\log n)$ time a candidate solution (i.e., a radius and a center) such that if that configuration happens then the candidate solution is an optimal solution. On the other hand, each solution is feasible in the sense that the distances from the candidate center to all vertices in $G(i,j)$ is no more than the candidate radius. After the candidate solutions for all (a constant number of) configurations are computed, we return the solution with the smallest radius. The details are given below.

If $c$ is on $P(v_1,v_i)$, then $d_P(v_1,v_i)\geq \alpha(i,j)$ and $r=(d_P(v_1,v_i)+\alpha(i,j))/2$. Correspondingly, for this configuration, our algorithm works as follows.

We first compute $d_P(v_1,v_i)$ and $\alpha(i,j)$. For $d_P(v_1,v_i)$, it can be easily computed in $O(1)$ time. For $\alpha(i,j)$, recall that $\alpha(i,j)=\max\{\beta(i,j),\gamma(i,j)\}$, $\beta(i,j)=\max_{k\in
[i,j]}d_{C(i,j)}(v_{i},v_k)$, and $\gamma(i,j)=|v_iv_j|+d_P(v_j,v_n)$. Clearly, $\gamma(i,j)$ can be obtained in $O(1)$ time. For $\beta(i,j)$, we can compute it in $O(\log n)$ time by binary search. Indeed, let $q$ be the point on $C(i,j)$ such that $d_{C(i,j)}(v_i,q)=(d_P(v_i,v_j)+|v_iv_j|)/2$, i.e., half of the total length of the cycle $C(i,j)$. Since $|v_iv_j|\leq d_P(v_i,v_j)$, $q$ must be on $P(v_i,v_j)$. Note that $q$ can be found in $O(\log n)$ time. Then, it can be verified the following is true. If $q$ is at a vertex of $P$, then $\beta(i,j)=(d_P(v_i,v_j)+|v_iv_j|)/2$. Otherwise, suppose $q$ is in the interior of the edge $e(v_k,v_{k+1})$ for some $k\in [i,j-1]$ ($k$ can be determined in the above binary search procedure for computing $q$); then $\beta(i,j)=\max\{d_P(v_i,v_k),|v_iv_j|+d_P(v_{k+1},v_j)\}$, which can be computed in $O(1)$ time.

After $d_P(v_1,v_i)$ and $\alpha(i,j)$ are computed, if $d_P(v_1,v_i)\geq \alpha(i,j)$, then we report $r=(d_P(v_1,v_i)+\alpha(i,j))/2$ and the point on $P(v_1,v_i)$ whose distance from $v_1$ is equal to $r$ as the center $c$. Note that $c$ can be found by binary search in $O(\log n)$ time. This finishes our algorithm for the case where $c$ is on $P(v_1,v_i)$. The algorithm runs in $O(\log n)$ time.

If $c$ is on $P(v_j,v_n)$, then we use a symmetric algorithm and we omit the details.

It remains to consider the configuration where $c\in C(i,j)\setminus\{v_i,v_j\}$. In the following, we consider the cases corresponding to those in Section~\ref{sec:algo}. As in Section~\ref{sec:algo}, whenever we discuss a configuration, we assume that none of the previously discussed configurations has an optimal solution.

\subsubsection{Case 1: $a=1$.}
In this case, depending on whether $b=n$ or $b\in [i,j]$, there are two subcases.

\subsubsection{Case 1.1: $b=n$.}
If $\pi$ does not contain $e(i,j)$, then $\pi$ is the entire path $P$. Correspondingly, the radius of our candidate solution is $P(v_1,v_n)/2$ and the center can be computed by binary search in $O(\log n)$ time.

If $\pi$ contains $e(i,j)$, since $c\in C(i,j)\setminus\{v_i,v_j\}$ and $c\in \pi$, $c$ must be in the interior of $e(i,j)$. Thus,  $r=(d_P(v_1,v_i)+|v_iv_j|+d_P(v_j,n))/2$ and $c$ is the point on $e(i,j)$ such that $d_P(v_1,v_i)+|v_ic|=r$.
In addition, it must hold that $d_{G}(c,v_k)\leq r$ for all $k\in [i+1,j-1]$. Hence, if $i'$ is the largest index such that $d_P(v_i,v_{i'})\leq d_P(v_1,v_i)$ and if $j'$ is the smallest index such that $d_P(v_j,v_{j'})\leq d_P(v_j,v_n)$, then it must hold that $j'\leq i'+1$.

Correspondingly, our algorithm works as follows. We compute $r'=(d_P(v_1,v_i)+|v_iv_j|+d_P(v_j,n))/2$. We also compute the two indices $i'$ and $j'$ as defined above, which can be done in $O(\log n)$ time. If $r'> d_P(v_1,v_i)$, $r'> d_P(v_j,v_n)$, and $j'\leq i'+1$, then we keep $r'$ as a candidate radius and $c'$ as the candidate center, where $c'$ is the point on $e(i,j)$ with $d_P(v_1,v_i)+|v_ic'|=r'$.
One can verify that our solution is feasible.

\subsubsection{Case 1.2: $b\in [i,j]$.}
Depending on whether $\pi$ contains $e(v_i,v_j)$, there are two cases.

\subsubsection{Case 1.2.1: $e(v_i,v_j)\in \pi$.}

Depending on whether $c$ is on $e(v_i,v_j)$, there are further two subcases.

\subsubsection{Case 1.2.1.1: $c\in e(v_i,v_j)$ (e.g., similar to Fig.~\ref{fig:config20}).}
More precisely, $c$ is in the interior of
$e(v_i,v_j)$.
Define $i'$ as the smallest index such that $d_P(v_1,v_i)<d_P(v_i,v_{i'})$. Since $v_1$ and $v_b$ are two farthest vertices of $c$ and $v_b\in P(v_i,v_{j})$, it can be verified that $b=i'$ (e.g., by a similar argument as for Lemma~\ref{lem:70}(2)) and $d_P(v_{i'},v_j)\geq d_P(v_j,v_n)$.
Further, $r=(d_P(v_1,v_i)+|v_iv_j|+d_P(v_{i'},v_j))/2$.

Correspondingly, our algorithm works as follows. We first compute $i'$
by binary search in $O(\log n)$ time. If such an index $i'$ does
not exist or if $i'\geq j$, then we stop the algorithm.
Otherwise, we compute $r'=(d_P(v_1,v_i)+|v_iv_j|+d_P(v_{i'},v_j))/2$.
If $d_P(v_1,v_i)<r'$ and $d_P(v_{i'},v_j)<r'$ (this makes sure that the
center is in the interior of $e(v_i,v_j)$) and
$d_P(v_{i'},v_j)\geq d_P(v_j,v_n)$, then we keep $r'$ as a candidate
radius for this case (the center can be computed on $e(v_i,v_j)$ accordingly).
One can verify that our solution is feasible.

\subsubsection{Case 1.2.1.2: $c\not\in e(v_i,v_j)$.}

In this case, $c$ is in $P(v_b,v_j)\setminus\{v_j\}$ (e.g., similar to the bottom example in Fig.~\ref{fig:config}). Also, $r=(d_P(v_1,v_i)+|v_iv_j|+d_P(v_{j},v_b))/2$ and $d_P(v_j,v_n)\leq d_P(v_1,v_i)+|v_iv_j|$.
Define the index $i'$ in the same way as in the above Case 1.2.1.1. Then, we also have $b=i'$ (by a similar argument as for Lemma~\ref{lem:75}(2)).

Correspondingly, our algorithm works as follows. We first check whether $d_P(v_j,v_n)\leq d_P(v_1,v_i)+|v_iv_j|$ is true. If not, then we stop our algorithm. Otherwise, we compute the
index $i'$ as defined in the above Case 1.2.1.1.
If such an index $i'$ does not exist or if $i'\geq j$, then we stop the algorithm.
Otherwise, we compute $r'=(d_P(v_1,v_i)+|v_iv_j|+d_P(v_j,v_{i'}))/2$.
If $d_P(v_{i'},v_j)>r'$, then we keep $r'$ as a
candidate radius (the candidate center can be computed accordingly on $P(v_{i'},v_j)$).
One can verify that our solution is feasible.

\subsubsection{Case 1.2.2: $e(v_i,v_j)\not\in \pi$.}

In this case, $c\in P(v_i,v_b)$, $\pi=P(v_1,c)\cup P(c,v_b)$ (similar to Fig.~\ref{fig:config122}), and
$r=d_P(v_1,v_b)/2$.
Let $i'$ be largest index in $[i,j]$ such that $d_P(v_1,v_i)< |v_iv_j|+d_P(v_{i'},v_j)$. One can verify that $b=i'$ (by a similar argument as for Lemma~\ref{lem:80}(2)).
Since $v_b$ is a farthest vertex and $d_G(c,v_b)=d_P(c,v_b)$,
for any $k\in [b+1,n]$, $d_G(c,v_{k})=d_P(c,v_i)+|v_iv_j|+d_P(v_j,v_{k})$.
In particular, we have
$d_G(c,v_{n})=d_P(c,v_i)+|v_iv_j|+d_P(v_j,v_{n})$.
Since $v_1$ is a farthest vertex and
$d_G(c,v_1)=d_P(c,v_1)=d_P(c,v_i)+d_P(v_i,v_1)$,
we can obtain $d_P(v_1,v_i)\geq |v_iv_j|+d_P(v_j,v_n)$.

Correspondingly, our algorithm works as follows. We first check whether $d_P(v_1,v_i)\geq |v_iv_j|+d_P(v_j,v_n)$. If not, we stop the algorithm. Otherwise, we compute the index $i'$, which can be done in $O(\log n)$ time by binary search. If no such index exists, then we stop the algorithm. Otherwise, we compute $r'=d_P(v_1,v_{i'})/2$. If $r'>d_P(v_1,v_i)$, then we keep $r'$ as a candidate radius (the  center can be determined correspondingly on $P(v_i,v_{i'})$).
One can verify that our solution is feasible.

\subsubsection{Case 2: $b=n$.}

This case is symmetric to the above Case 1 and we omit the details.

\subsubsection{Case 3: $a,b\in [i,j]$.}

Depending on whether $c\in e(v_i,v_j)$, there are two cases.

\subsubsection{Case 3.1: $c\in e(v_i,v_j)$.}

In this case, as argued in Section~\ref{sec:algo}, $b=a+1$ (similar to the example in Fig.~\ref{fig:config31}). Let $|C(i,j)|$ denote the total length of the cycle $C(i,j)$. We have $r=(d_P(v_i,v_a)+|v_iv_j|+d_P(v_b,v_j))/2=(|C(i,j)|-d_P(v_a,v_b))/2$.

Let $i'$ be the smallest index such that $d_P(v_1,v_i)\leq d_P(v_i,v_{i'})$ and
let $j'$ be the largest index such that $d_P(v_{j'},v_j)\geq d_P(v_j,v_{n})$.
Let $i''$ be the smallest index such that $|v_jv_i|+d_P(v_i,v_{i''})+d_P(v_{i''},v_{i''+1})/2>|C(i,j)|/2$ (intuitively, if $p$ is the point on $C(i,j)$ whose distance from $v_j$ is $|C(i,j)|/2$, then if we traverse $C(i,j)$ from $v_j$ in counterclockwise order, after we pass $p$, the mid point of $d_P(v_{i''},v_{i''+1})$ will be encountered first among the mid points of $P(v_k,v_{k+1})$ for all $k\in [i,j-1]$; e.g., see Fig.~\ref{fig:midpoint}).
Symmetrically, let $j''$ be the largest index such that $|v_iv_j|+d_P(v_j,v_{j''-1})+d_P(v_{j''-1},v_{j''})/2>|C(i,j)|/2$.
Let $a'=\max\{i',i''\}$ and $b'=\min\{j',j''\}$.
We have the following lemma.

\begin{figure}[t]
\begin{minipage}[t]{\textwidth}
\begin{center}
\includegraphics[height=1.7in]{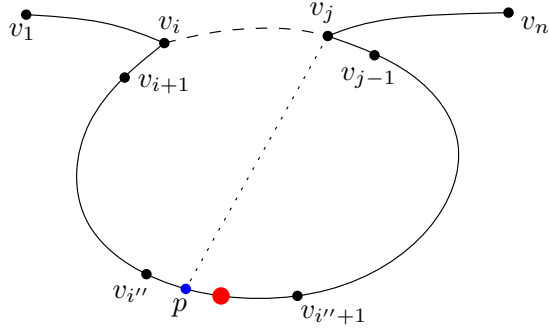}
\caption{\footnotesize  Illustrating the definition of $i''$: the dashed curve is $e(v_i,v_j)$, $p$ is the point whose distance from $v_j$ is equal to $|C(i,j)|/2$, and the (red) large point between $p$ and $v_{i''+1}$ is the mid point of $P(v_{i''},v_{i''+1})$.}
\label{fig:midpoint}
\end{center}
\end{minipage}
\vspace{-0.15in}
\end{figure}

\begin{lemma}
$a'\leq a< b\leq b'$.
$a$ is the index $k\in [a',b'-1]$ such that $d_P(v_k,v_{k+1})$ is the largest.
\end{lemma}
\begin{proof}
Clearly, $d_G(c,v_1)=|cv_i|+d_P(v_i,v_1)$ and $d_G(c,v_a)=|cv_i|+d_P(v_i,v_a)$. Since $d_G(c,v_1)\leq d_G(c,v_a)=r$, $d_P(v_1,v_i)\leq d_P(v_i,v_a)$. The definition of $i'$ implies that $a\geq i'$. By a similar argument we can show $b\leq j'$.

On the other hand, since $c$ is in the interior of $e(v_i,v_j)$, $|v_iv_j|> |cv_i|$ and
$|cv_i|+d_P(v_i,v_{a}) = r = (|C(i,j)|-d_P(v_a,v_{a+1}))/2$. Hence,
$|cv_i|+d_P(v_i,v_{a})+d_P(v_{a},v_{a+1})/2= |C(i,j)|/2$. Also, since $|v_iv_j|> |cv_i|$, we obtain
$|v_jv_i|+d_P(v_i,v_{a})+d_P(v_{a},v_{a+1})> |C(i,j)|/2$.
According to the definition of $i''$, we have $a\geq i''$. By a similar argument we can show $b\leq j''$.

The above proves that $a'\leq a<b\leq b'$.

For the second statement of the lemma,
assume to the contrary that $d_P(v_a,v_{a+1})<d_P(v_k,v_{k+1})$ for some $k\neq a$ and $k\in [a',b'-1]$.
Then, let $c'$ be the point on $C(i,j)$ whose distance from the mid point of $P(v_k,v_{k+1})$ is $|C(i,j)|/2$. Let $r'=(|C(i,j)|-d_P(v_k,v_{k+1}))/2$, which is smaller than $r$ due to $d_P(v_a,v_{a+1})<d_P(v_k,v_{k+1})$. Since $k\in [a',b'-1]$, one can verify that $c'$ is in the interior of $e(v_i,v_j)$, $|c'v_i|+d_P(v_i,v_1)\leq r'$, and $|c'v_j|+d_P(v_j,v_n)\leq r'$. Therefore, the distances from $c'$ to all vertices in $G(i,j)$ are smaller than $r$, which contradicts with that $r$ is the radius of $G(i,j)$. \qed
\end{proof}

Correspondingly, our algorithm works as follows. We first compute $a'$ and $b'$, which can be done by binary search in $O(\log n)$ time. If one of $a'$ and $b'$ does not exist, or $a'\geq b'$, then we stop the algorithm. Otherwise, we compute the index $k\in [a',b'-1]$ such that $d_P(v_k,v_{k+1})$ is the largest, which can be done in $O(\log n)$ time by our range-maxima data structure using the query range $(a',b')$. Then, we compute $r'=(|C(i,j)|-d_P(v_k,v_{k+1}))/2$ as the radius for our candidate solution (the center can be determined on $e(v_i,v_j)$ accordingly).
One can verify that our solution is feasible.

\subsubsection{Case 3.2: $c\not\in e(v_i,v_j)$.}

In this case, $c$ is in the interior of $P(v_i,v_j)$. Similar to the argument in Section~\ref{sec:algo}, $b=a+1$, and $c$ is either to the right of $v_b$ or to the left of $v_a$. We only discuss the former case since the latter case is similar. In this case (similar to the example in Fig.~\ref{fig:config32}), for each $k\in [1,a]$, the shortest path from $c$ to $v_k$ in $G$ contains $e(v_i,v_j)$. For each each $k\in [b,n]$, the shortest path from $c$ to $v_k$ is $P(c,v_k)$. Since $v_a$ is a farthest vertex, we have $d_P(v_1,v_i)\leq d_P(v_i,v_a)$ and $d_P(v_j,v_n)\leq |v_jv_i|+d_P(v_i,v_a)$.

Let $i'$ be the smallest index such that $d_P(v_1,v_i)\leq d_P(v_i,v_{i'})$.
Let $i''$ be the smallest index such that $d_P(v_j,v_n)\leq |v_jv_i|+d_P(v_i,v_{i''})$.
Let $b'$ be the largest index such that $d_P(v_j,v_{b'})+d_P(v_{b'},v_{b'-1})/2>|C(i,j)|/2$.
Let $a'=\max\{i',i''\}$. We have the following lemma.

\begin{lemma}
$a'\leq a< b\leq b'$.
$a$ is the index $k\in [a',b'-1]$ such that $d_P(v_k,v_{k+1})$ is the largest.
\end{lemma}
\begin{proof}
	Since $d_P(v_1,v_i)\leq d_P(v_i,v_a)$ and $d_P(v_j,v_n)\leq |v_iv_j|+d_P(v_i,v_a)$,
we obtain that $a\geq i'$ and $a\geq i''$, and thus $a\geq a'$.

	On the other hand, since $c$ is in $P(v_b,v_j)\setminus\{v_j\}$, $d_P(v_j,v_b)>d_P(c,v_b)=r=(|C(i,j)|-d_P(v_b,v_{b-1}))/2$. Hence,
	$d_P(v_{j},v_b)+d_P(v_{b},v_{b-1})/2 >|C(i,j)|/2$.
	Therefore, $b\leq b'$.

The above proves that $a'\leq a<b\leq b'$.

For the second statement of the lemma, assume to the contrary that $d_P(v_a,v_{a+1})<d_P(v_k,v_{k+1})$ for some $k\neq a$ and $k\in [a',b'-1]$.
Then, let $c'$ be the point in $C(i,j)$ whose distance from the mid point of $P(v_k,v_{k+1})$ is $|C(i,j)|/2$. Let	$r'=(|C(i,j)|-d_P(v_k,v_{k+1}))/2$, which is smaller than $r$ due to $d_P(v_a,v_{a+1})<d_P(v_k,v_{k+1})$.
	Since $k\in [a',b'-1]$, one can verify that $c'$ is in
	$P(v_{k+1},v_j)\setminus\{v_j\}$ and $d_P(v_1,v_i)+|v_iv_j|+d_P(v_j,c)\leq r'$ and $d_P(c',v_n)\leq r'$. Therefore, the distances from $c'$ to all vertices in $G(i,j)$ are
	smaller than $r$, which incurs contradiction. \qed
\end{proof}

Correspondingly, our algorithm works as follows. We first compute $a'$
and $b'$  in $O(\log n)$ time by binary search. If one of $a'$ and $b'$ does not exist, or $a'\geq b'$, then
we stop the algorithm. Otherwise, we compute the index $k\in
[a',b'-1]$ such that $d_P(v_k,v_{k+1})$ is the largest, which can be
done in $O(\log n)$ time by using a range-maxima query as in Case 3.1. Finally, we compute $r'=(|C(i,j)|-d_P(v_k,v_{k+1}))/2$ as the radius for our candidate solution (the center can be determined on $P(v_{k+1},v_j)$ accordingly).
One can verify that our solution is feasible.

\paragraph{Summary.} The above gives our algorithms for all
possible configurations. For each configuration, we compute a candidate solution (a radius $r'$ and a center $c'$) in $O(\log n)$ time, such that if the configuration has an optimal solution then our solution is also optimal. On the other hand, each candidate solution is feasible in the sense that the distances from $c'$ to all vertices in $G(i,j)$ are at most $r'$. Among all (a constant number of) candidate solutions, we return the one with the smallest radius. The running time is $O(\log n)$.

\begin{theorem}
	After $O(n)$ time preprocessing, given any query index pair
	$(i,j)$, we can compute the radius and the center of the augmented
	graph $P\cup\{e(v_i,v_j)\}$ in $O(\log n)$ time.
\end{theorem}


%

\bibliographystyle{plain}
\bibliography{reference}





\end{document}